\documentclass[journal, onecolumn, 12pt]{IEEEtran}

\usepackage{xr-hyper}

\usepackage[T1]{fontenc}
\usepackage[utf8]{inputenc}
\usepackage{amsmath,amssymb,amsfonts,mathrsfs,bm}
\usepackage{mathtools}
\usepackage{amsthm}
\usepackage{scalerel}
\usepackage{nicefrac}
\usepackage[shortlabels]{enumitem}
\usepackage{graphicx}
\usepackage{epstopdf}
\DeclareGraphicsExtensions{.eps,.png,.jpg,.pdf}

\usepackage{url}
\usepackage{colortbl}
\usepackage{booktabs}
\usepackage{multirow}
\usepackage[table,dvipsnames]{xcolor}
\usepackage[normalem]{ulem}
\usepackage{xparse}
\usepackage{calc}
\usepackage{etoolbox}

\makeatletter
\@ifpackageloaded{natbib}{
	\relax
}{
	\usepackage{cite}
}
\makeatother

\usepackage{array}
\newcolumntype{L}[1]{>{\raggedright\let\newline\\\arraybackslash\hspace{0pt}}m{#1}}
\newcolumntype{C}[1]{>{\centering\let\newline\\\arraybackslash\hspace{0pt}}m{#1}}
\newcolumntype{R}[1]{>{\raggedleft\let\newline\\\arraybackslash\hspace{0pt}}m{#1}}

\makeatletter
\let\MYcaption\@makecaption
\makeatother
\usepackage[font=footnotesize]{subcaption}
\makeatletter
\let\@makecaption\MYcaption
\makeatother

\usepackage{glossaries}
\makeatletter
\let\oldgls\gls
\let\oldglspl\glspl

\newcommand\fussy@ifnextchar[3]{%
	\let\reserved@d=#1%
	\def\reserved@a{#2}%
	\def\reserved@b{#3}%
	\futurelet\@let@token\fussy@ifnch}
\def\fussy@ifnch{%
	\ifx\@let@token\reserved@d
		\let\reserved@c\reserved@a
	\else
		\let\reserved@c\reserved@b
	\fi
	\reserved@c}

\renewcommand{\gls}[1]{%
\oldgls{#1}\fussy@ifnextchar.{\@checkperiod}{\@}}
\renewcommand{\glspl}[1]{%
\oldglspl{#1}\fussy@ifnextchar.{\@checkperiod}{\@}}

\newcommand{\@checkperiod}[1]{%
	\ifnum\sfcode`\.=\spacefactor\else#1\fi
}
\makeatother

\newacronym{wrt}{w.r.t.}{with respect to}
\newacronym{RHS}{R.H.S.}{right-hand side}
\newacronym{LHS}{L.H.S.}{left-hand side}
\newacronym{iid}{i.i.d.}{independent and identically distributed}
\newacronym{SVD}{SVD}{singular value decomposition}
\newacronym{5G}{5G}{fifth generation wireless}
\newacronym{3GPP}{3GPP}{third generation partnership project}
\newacronym{OFDM}{OFDM}{orthogonal frequency-division multiplexing}

\usepackage{float}

\ifx\notloadhyperref\undefined
	\ifx\loadbibentry\undefined
		\usepackage[hidelinks,hypertexnames=false]{hyperref}
	\else
		\usepackage{bibentry}
		\makeatletter\let\saved@bibitem\@bibitem\makeatother
		\usepackage[hidelinks,hypertexnames=false]{hyperref}
		\makeatletter\let\@bibitem\saved@bibitem\makeatother
	\fi
\else
	\ifx\loadbibentry\undefined
		\relax
	\else
		\usepackage{bibentry}
	\fi
\fi

\usepackage[capitalize]{cleveref}
\crefname{equation}{}{}
\Crefname{equation}{}{}
\crefname{claim}{claim}{claims}
\crefname{step}{step}{steps}
\crefname{line}{line}{lines}
\crefname{condition}{condition}{conditions}
\crefname{dmath}{}{}
\crefname{dseries}{}{}
\crefname{dgroup}{}{}

\crefname{Problem}{Problem}{Problems}
\crefformat{Problem}{Problem~(#2#1#3)}
\crefrangeformat{Problem}{Problems~(#3#1#4) to~(#5#2#6)}

\crefname{Theorem}{Theorem}{Theorems}
\crefname{Corollary}{Corollary}{Corollaries}
\crefname{Proposition}{Proposition}{Propositions}
\crefname{Lemma}{Lemma}{Lemmas}
\crefname{Definition}{Definition}{Definitions}
\crefname{Example}{Example}{Examples}
\crefname{Assumption}{Assumption}{Assumptions}
\crefname{Remark}{Remark}{Remarks}
\crefname{Rem}{Remark}{Remarks}
\crefname{remarks}{Remarks}{Remarks}
\crefname{Appendix}{Appendix}{Appendices}
\crefname{Supplement}{Supplement}{Supplements}
\crefname{Exercise}{Exercise}{Exercises}
\crefname{Theorem_A}{Theorem}{Theorems}
\crefname{Corollary_A}{Corollary}{Corollaries}
\crefname{Proposition_A}{Proposition}{Propositions}
\crefname{Lemma_A}{Lemma}{Lemmas}
\crefname{Definition_A}{Definition}{Definitions}

\usepackage{crossreftools}
\ifx\notloadhyperref\undefined
\else
	\relax
\fi

\usepackage{algorithm,algorithmic}

\ifx\loadbreqn\undefined
	\relax
\else
	\usepackage{breqn}
\fi

\makeatletter
\def\cleartheorem#1{%
    \expandafter\let\csname#1\endcsname\relax
    \expandafter\let\csname c@#1\endcsname\relax
}
\def\clearthms#1{ \@for\tname:=#1\do{\cleartheorem\tname} }
\makeatother

\ifx\renewtheorem\undefined
	\ifx\useTheoremCounter\undefined
		\newtheorem{Theorem}{Theorem}
		\newtheorem{Corollary}{Corollary}
		\newtheorem{Proposition}{Proposition}
		\newtheorem{Lemma}{Lemma}
	\else
		\newtheorem{Theorem}{Theorem}
		\newtheorem{Corollary}[Theorem]{Corollary}
		\newtheorem{Proposition}[Theorem]{Proposition}
	\fi

	\newtheorem{Remark}{Remark}

\fi

\theoremstyle{remark}

\theoremstyle{plain}

\newcommand{\qednew}{\nobreak \ifvmode \relax \else
		\ifdim\lastskip<1.5em \hskip-\lastskip
			\hskip1.5em plus0em minus0.5em \fi \nobreak
		\vrule height0.75em width0.5em depth0.25em\fi}



\NewDocumentCommand{\movedownsub}{e{^_}}{%
	\IfNoValueTF{#1}{%
		\IfNoValueF{#2}{^{}}
	}{%
		^{#1}
	}%
	\IfNoValueF{#2}{_{#2}}
}

\let\latexchi\chi
\RenewDocumentCommand{\chi}{}{\latexchi\movedownsub}

\newcommand{\ba}{\mathbf{a}}
\newcommand{\bA}{\mathbf{A}}

\newcommand{\bC}{\mathbf{C}}
\newcommand{\bd}{\mathbf{d}}

\newcommand{\bI}{\mathbf{I}}

\newcommand{\bN}{\mathbf{N}}

\newcommand{\bP}{\mathbf{P}}

\newcommand{\bR}{\mathbf{R}}

\newcommand{\bx}{\mathbf{x}}
\newcommand{\bX}{\mathbf{X}}

\newcommand{\bY}{\mathbf{Y}}

\DeclareSymbolFont{bsfletters}{OT1}{cmss}{bx}{n}
\DeclareSymbolFont{ssfletters}{OT1}{cmss}{m}{n}
\DeclareMathSymbol{\bsfGamma}{0}{bsfletters}{'000}
\DeclareMathSymbol{\ssfGamma}{0}{ssfletters}{'000}
\DeclareMathSymbol{\bsfDelta}{0}{bsfletters}{'001}
\DeclareMathSymbol{\ssfDelta}{0}{ssfletters}{'001}
\DeclareMathSymbol{\bsfTheta}{0}{bsfletters}{'002}
\DeclareMathSymbol{\ssfTheta}{0}{ssfletters}{'002}
\DeclareMathSymbol{\bsfLambda}{0}{bsfletters}{'003}
\DeclareMathSymbol{\ssfLambda}{0}{ssfletters}{'003}
\DeclareMathSymbol{\bsfXi}{0}{bsfletters}{'004}
\DeclareMathSymbol{\ssfXi}{0}{ssfletters}{'004}
\DeclareMathSymbol{\bsfPi}{0}{bsfletters}{'005}
\DeclareMathSymbol{\ssfPi}{0}{ssfletters}{'005}
\DeclareMathSymbol{\bsfSigma}{0}{bsfletters}{'006}
\DeclareMathSymbol{\ssfSigma}{0}{ssfletters}{'006}
\DeclareMathSymbol{\bsfUpsilon}{0}{bsfletters}{'007}
\DeclareMathSymbol{\ssfUpsilon}{0}{ssfletters}{'007}
\DeclareMathSymbol{\bsfPhi}{0}{bsfletters}{'010}
\DeclareMathSymbol{\ssfPhi}{0}{ssfletters}{'010}
\DeclareMathSymbol{\bsfPsi}{0}{bsfletters}{'011}
\DeclareMathSymbol{\ssfPsi}{0}{ssfletters}{'011}
\DeclareMathSymbol{\bsfOmega}{0}{bsfletters}{'012}
\DeclareMathSymbol{\ssfOmega}{0}{ssfletters}{'012}

\newcommand{\bPhi}{\bm{\Phi}}

\makeatletter
\newcommand*\rel@kern[1]{\kern#1\dimexpr\macc@kerna}
\newcommand*\widebar[1]{%
  \begingroup
  \def\mathaccent##1##2{%
    \rel@kern{0.8}%
    \overline{\rel@kern{-0.8}\macc@nucleus\rel@kern{0.2}}%
    \rel@kern{-0.2}%
  }%
  \macc@depth\@ne
  \let\math@bgroup\@empty \let\math@egroup\macc@set@skewchar
  \mathsurround\z@ \frozen@everymath{\mathgroup\macc@group\relax}%
  \macc@set@skewchar\relax
  \let\mathaccentV\macc@nested@a
  \macc@nested@a\relax111{#1}%
  \endgroup
}
\makeatother

\DeclareMathOperator*{\argmax}{arg\,max}

\DeclareMathOperator{\rank}{rank}

\newcommand{\ifbcdot}[1]{\ifblank{#1}{\cdot}{#1}}

\DeclarePairedDelimiterX\abs[1]{\lvert}{\rvert}{\ifbcdot{#1}}
\DeclarePairedDelimiterX\parens[1]{(}{)}{\ifbcdot{#1}}
\DeclarePairedDelimiterX\brk[1]{[}{]}{\ifbcdot{#1}}
\DeclarePairedDelimiterX\braces[1]{\{}{\}}{\ifbcdot{#1}}
\DeclarePairedDelimiterX\angles[1]{\langle}{\rangle}{\ifblank{#1}{\cdot,\cdot}{#1}}
\DeclarePairedDelimiterX\ip[2]{\langle}{\rangle}{\ifbcdot{#1},\ifbcdot{#2}}
\DeclarePairedDelimiterX\norm[1]{\lVert}{\rVert}{\ifbcdot{#1}}
\DeclarePairedDelimiterX\ceil[1]{\lceil}{\rceil}{\ifbcdot{#1}}
\DeclarePairedDelimiterX\floor[1]{\lfloor}{\rfloor}{\ifbcdot{#1}}

\DeclarePairedDelimiterXPP\trace[1]{\operatorname{Tr}}{(}{)}{}{\ifbcdot{#1}} 
\DeclarePairedDelimiterXPP\col[1]{\operatorname{col}}{\{}{\}}{}{\ifbcdot{#1}} 
\DeclarePairedDelimiterXPP\row[1]{\operatorname{row}}{\{}{\}}{}{\ifbcdot{#1}} 
\DeclarePairedDelimiterXPP\erf[1]{\operatorname{erf}}{(}{)}{}{\ifbcdot{#1}}
\DeclarePairedDelimiterXPP\erfc[1]{\operatorname{erfc}}{(}{)}{}{\ifbcdot{#1}}
\DeclarePairedDelimiterXPP\KLD[2]{D}{(}{)}{}{\ifbcdot{#1}\, \delimsize\|\, \ifbcdot{#2}} 
\DeclarePairedDelimiterXPP\op[2]{\operatorname{#1}}{(}{)}{}{#2} 

\DeclarePairedDelimiterXPP\indicate[1]{{\bf 1}}{\{}{\}}{}{\ifbcdot{#1}}

\providecommand\given{}

\DeclarePairedDelimiterX\Set[2]\{\}{%
\renewcommand\given{\SetSymbol[\delimsize]{#1}}
#2
}
\DeclarePairedDelimiterX\Setc[1]\{\}{%
\renewcommand\given{\SetSymbol{:}}
#1
}

\NewDocumentCommand\set{s o m}{%
	\IfBooleanTF#1%
	{\IfValueTF{#2}{\Set*{#2}{#3}}{\Setc*{#3}}}%
	{\IfValueTF{#2}{\Set{#2}{#3}}{\Setc{#3}}}%
}

\NewDocumentCommand{\evalat}{ s O{\big} m e{_^} }{%
\IfBooleanTF{#1}%
{\left. #3 \right|}{#3#2|}%
\IfValueT{#4}{_{#4}}%
\IfValueT{#5}{^{#5}}%
}

\NewDocumentCommand \ifcondp {m m} {%
	#1%
	\IfValueT{#2}{\given #2}%
}

\providecommand\given{}
\DeclarePairedDelimiterXPP\cprob[1]{}(){}{
\renewcommand\given{\nonscript\,\delimsize\vert\allowbreak\nonscript\,\mathopen{}}
\ifcondp#1
}
\DeclarePairedDelimiterXPP\cexp[1]{}[]{}{
\renewcommand\given{\nonscript\,\delimsize\vert\allowbreak\nonscript\,\mathopen{}}
\ifcondp#1
}

\DeclareDocumentCommand \P { s e{_^} >{\SplitArgument{ 1 }{ @| }}d() g } {%
	\mathbb{P}%
	\IfBooleanTF{#1}%
		{
			\IfValueT{#2}{_{#2}}%
			\IfValueT{#3}{^{#3}}%
			\IfValueTF{#5}%
				{\cprob{#4 \given #5}}%
				{\IfValueT{#4}{\cprob{#4}}}%
		}%
		{
			\IfValueT{#2}{_{#2}}%
			\IfValueT{#3}{^{#3}}%
			\IfValueTF{#5}%
				{\cprob*{#4 \given #5}}%
				{\IfValueT{#4}{\cprob*{#4}}}%
		}%
}

\DeclareDocumentCommand \E { s e{_^} >{\SplitArgument{ 1 }{ @| }}d[] g } {%
	\mathbb{E}%
	\IfBooleanTF{#1}%
		{
			\IfValueT{#2}{_{#2}}%
			\IfValueT{#3}{^{#3}}%
			\IfValueTF{#5}%
				{\cexp{#4 \given #5}}%
				{\IfValueT{#4}{\cexp{#4}}}%
		}%
		{
			\IfValueT{#2}{_{#2}}%
			\IfValueT{#3}{^{#3}}%
			\IfValueTF{#5}%
				{\cexp*{#4 \given #5}}%
				{\IfValueT{#4}{\cexp*{#4}}}%
		}%
}

\ExplSyntaxOn
\NewDocumentCommand \dist {m o o} {%
\mathrm{#1}\left(%
	\IfValueT{#3}{%
		\tl_if_blank:nTF{ #3 }{\cdot\, \middle|\, }{#3\, \middle|\, }%
	}
	\IfValueT{#2}{#2}%
\right)%
}
\ExplSyntaxOff

\NewDocumentCommand {\cbrace} {t+ D[]{black} D(){\widthof{#5}} m m } {%
	\begingroup%
		\color{#2}
		\IfBooleanTF{#1}{%
			\overbrace{#4}^%
		}{
			\underbrace{#4}_%
		}%
		{\parbox[c]{#3}{\centering\footnotesize{#5}}}%
	\endgroup%
}

\let\oldforall\forall
\renewcommand{\forall}{\oldforall \, }

\let\oldexist\exists
\renewcommand{\exists}{\oldexist \, }

\graphicspath{{./Figures/}{./figures/}}
\DeclareDocumentCommand{\includeCroppedPdf}{ o O{./Figures/} m }{
	\IfFileExists{#2#3-crop.pdf}{}{%
		\immediate\write18{pdfcrop #2#3.pdf #2#3-crop.pdf}}%
	\includegraphics[#1]{#2#3-crop.pdf}
}

\makeatletter
\newcommand*{\addFileDependency}[1]{
  \typeout{(#1)}
  \@addtofilelist{#1}
  \IfFileExists{#1}{}{\typeout{No file #1.}}
}
\makeatother

\definecolor{gray90}{gray}{0.9}

\ifx\nohighlights\undefined

	\newcommand{\msout}[1]{\text{\color{green} \sout{\ensuremath{#1}}}}
	\newcommand{\del}[1]{{\color{green}\ifmmode \msout{#1}\else\sout{#1}\fi}}
\else

	\newcommand{\msout}[1]{#1}
	\newcommand{\del}[1]{#1}
\fi

\newcommand{\hhide}[1]{}

\ifx\diagnoselabel\undefined
	\relax
\else
	\makeatletter
	\def\@testdef #1#2#3{%
		\def\reserved@a{#3}\expandafter \ifx \csname #1@#2\endcsname
			\reserved@a  \else
			\typeout{^^Jlabel #2 changed:^^J%
				\meaning\reserved@a^^J%
				\expandafter\meaning\csname #1@#2\endcsname^^J}%
			\@tempswatrue \fi}
	\makeatother
\fi

  \def\R{{\mathbb{R}}}    \def\E{{\mathbb{E}}}

\newcommand{\beq}{\begin{eqnarray}}
\newcommand{\eeq}{\end{eqnarray}}

  \def\cC{{\mathcal{C}}}

 \def\cN{{\mathcal{N}}}

           \def\lA{\left\|}     \def\rA{\right\|}

\makeatletter
\renewenvironment{proof}[1][\proofname]{\par
  \pushQED{\qed}%
  \normalfont \topsep0\p@\relax
  \trivlist
  \item[\hskip3\labelsep\itshape#1\@addpunct{:}]\ignorespaces}{%
  \popQED\endtrivlist\@endpefalse
}
\makeatother

\newacronym{MLE}{MLE}{maximum likelihood estimate}

\usepackage{setspace}
\usepackage{autobreak}

\begin{document}

\pagenumbering{arabic}

\title{Successful Recovery Performance Guarantees of SOMP Under the $\ell_2$-Norm of Noise}
\author{Wei Zhang,~\IEEEmembership{Member,~IEEE} and Taejoon Kim,~\IEEEmembership{Senior Member,~IEEE}
\thanks{

{W. Zhang is with the School of Electronics and Information Engineering, Harbin Institute of Technology Shenzhen, 518055, China (e-mail: zhangwei.sz@hit.edu.cn).}


{T. Kim is with the Department of Electrical Engineering and Computer Science, The University of Kansas, KS 66045, USA (e-mail: taejoonkim@ku.edu).}
}
}
\maketitle
\begin{abstract}
The simultaneous orthogonal matching pursuit (SOMP) is a popular, greedy approach for common support recovery of a row-sparse matrix.
However, compared to the noiseless scenario, the performance analysis of noisy SOMP is still nascent, especially in the scenario of unbounded noise.
When the measurement matrix and sparse signal are deterministic, we present a new study based on the mutual incoherence property (MIP) for performance analysis of noisy SOMP.
Specifically, when noise is bounded, we provide the condition on which the exact support recovery is guaranteed in terms of the MIP.
When noise is unbounded, we instead derive a bound on the successful recovery probability (SRP) that depends on the specific distribution of the $\ell_2$-norm of the noise matrix.
Then we focus on the typical case when noise is random Gaussian, and show that the lower bound of SRP follows Tracy-Widom law distribution.
The analysis reveals the number of measurements, noise level, the number of sparse vectors, and the value of mutual coherence that are required to guarantee a predefined recovery performance.
Theoretically, we show that the mutual coherence of the measurement matrix must decrease proportionally to the noise standard deviation, and the number of sparse vectors needs to grow proportionally to the noise variance.
Finally, we extensively validate the derived analysis through numerical simulations.

\end{abstract}

\section{Introduction}
The problem of sparse signal recovery appears in various applications of wireless communications and image processing \cite{Donoho2006CS,Tropp2006Greedy, duanOMP,ZhangSD,Xiao2017Efficient,NoiseSOMP,OMPTT},
in which a common linear observation model is assumed
\begin{align}
\bY = \bPhi \bC + \bN, \label{obOMP}
\end{align}
where $\bY \in \R^{M \times d}$ is the observation, $ \bPhi \in \R^{M\times N}$ with $M \ll N$ is the measurement matrix, $\bC\in \R^{N\times d}$ is the row-sparse matrix with only $L\ll N$ rows being non-zero, and $\bN\in \R^{M \times d}$ is the measurement noise matrix.
Without loss of generality, we assume that each column of $\bPhi$ has unit $\ell_2$-norm.
Unlike the single measurement vector (SMV) scenario, in which $d=1$ in \eqref{obOMP},
the case with $d> 1$ is commonly referred to as multiple measurement vectors (MMV) model \cite{Tropp2006Greedy,ChenMMV}, where the columns of $\bC$ share the same support.


Given the model in \eqref{obOMP}, the goal is to recover the support set of $\bC$ from the observations $\bY$ and the known measurement matrix $\bPhi$.
The potential recovery performance of MMV can be studied from an information-theoretic point-of-view \cite{Jin2013Support}.
There have been multiple variants of support recovery algorithms, such as greedy approaches \cite{TroppOMP, OMPTT,ChenMMV,vanMMV,LeeMMV,gribonval2008atoms,ZhangMMV2022}, subspace methods \cite{Kim2012Music,ZhangSD,Zhang2020Sequential}, convex relaxation \cite{tibshirani1996regression}, and message passing algorithms \cite{donoho2009message,kim2015virtual}. Among these algorithms, the greedy ones can achieve a nearly optimal recovery performance with low complexity, such as the orthogonal matching pursuit (OMP) for SMV models ($d=1$) \cite{TroppOMP, OMPTT,duanOMP}, and simultaneous OMP (SOMP) for MMV models ($d>1$) \cite{ChenMMV,vanMMV,LeeMMV}.
Due to these advantages, the OMP and SOMP models are widely used in localization/positioning \cite{zhu2020millimeter,gao2022method,gligoric2018visible}, sparse channel estimation \cite{OMPchannel,You2022Struc,ZhangMMV2022}, and signal detection for wireless communication systems \cite{Xiao2017Efficient,qin2019transient}.

Both OMP and SOMP are iterative algorithms, in which one atom (one column of $\bPhi$) is selected per iteration and added to the recovered support set.
The performance of support recovery of OMP and SOMP has been analyzed for both  noiseless \cite{Tropp2006Greedy,LosOMP,OMPbound,sharpOMP} and noisy scenarios \cite{OMPTT,MuOMP,Mi2017Prob,ChangRIP,Ge2019RIP,Cai2018Improve,NoiseSOMP,Wu2013Exact,Wang2013Performance,gribonval2008atoms}, in which two important characteristics of measurement matrix $\bPhi$ are
widely adopted: (i) restricted isometry property (RIP) and (ii) mutual incoherence property (MIP).
{
Specifically, a matrix $\bPhi$ satisfies the RIP of order $L$ with the restricted isometry constant (RIC) $\delta_L$ if $\delta_L$ is the
smallest constant such that
\begin{align*}
(1-\delta_L)\| \bx\|_2^2 \le \| \bPhi\bx \|_2^2\le (1+\delta_L)\| \bx\|_2^2
\end{align*}
holds for all $L$-sparse signal $\bx \in \R^{N\times 1}$. Meanwhile, the mutual coherence of a matrix $\bPhi$ is defined as
\begin{align}
\mu = \max_{i\neq j} \left| \langle [\bPhi]_{:,i}, [\bPhi]_{:,j}\rangle \right|, \label{def of mu}
\end{align}
where $ [\bPhi]_{:,i}$ denotes the $i$th column of $\bPhi$, and $\langle \cdot ,\cdot \rangle$ denotes the inner product.
}

There are existing works analyzing the recovery guarantees of OMP \cite{sharpOMP,MuOMP,OMPTT,Mi2017Prob,OMPbound,ChangRIP,Ge2019RIP,Cai2018Improve,amirazTight2021} and SOMP \cite{Tropp2006Greedy,liImproved2018,liFundamental2019,NoiseSOMP,zhangResults2022}.
In the noiseless scenario, OMP can successfully recover the support set in $L$ iterations when the RIC or the mutual coherence of the measurement matrix $\bPhi$ satisfy the conditions $\delta_{L+1} < 1/\sqrt{L+1}$ \cite{sharpOMP} or $\mu < {1}/{(2L-1)}$ \cite{MuOMP}, respectively.
Regarding the recovery performance of noisy OMP, the sufficient conditions for exact support recovery have been studied in terms of the MIP \cite{OMPTT,Mi2017Prob,amirazTight2021} and the RIP \cite{OMPbound,ChangRIP,Ge2019RIP,Cai2018Improve}, which have the additional requirements for the magnitude of the signal compared to the noiseless scenario.
There are also existing works that have analyzed the recovery guarantee for both noiseless and noisy SOMP \cite{Tropp2006Greedy,liImproved2018,liFundamental2019,NoiseSOMP,zhangResults2022}.
In the noiseless case, SOMP guarantees exact support recovery when $\mu < {1}/{(2L-1)}$ \cite{Tropp2006Greedy} or $\delta_{L+1} < 1/\sqrt{L+1}$ \cite{liImproved2018,zhangResults2022}, aligning with the performance guarantees of noiseless OMP.
Regarding noisy SOMP, the recovery error probability under Gaussian noise was analyzed based on the RIP \cite{NoiseSOMP}, which reveals that achieving near-zero error is possible when the signal power and the number of sparse vectors are sufficiently large.
In \cite{liFundamental2019}, an additional condition on the signal magnitude necessary  is introduced in to guarantee successful recovery of SOMP in the noisy case,  in addition to the requirement for RIC. 
Furthermore, there are variants of SOMP proposed in \cite{yaoSparsity2018,liNew2019,kimJoint2020a}, with performance analyses provided based on RIP.
{We note the performance analyses for the noisy SOMP in \cite{NoiseSOMP,liFundamental2019} and the variants in \cite{yaoSparsity2018,liNew2019,kimJoint2020a} are all based on RIP. However, compared to RIP, the advantages of characterizing the mutual coherence in \eqref{def of mu} are in (i) the accountability in terms of being able to capture the property of maximal correlation between different columns of a measurement matrix and (ii) the amenability in terms of being able to evaluate effectively for a fixed measurement matrix.
There are prior works \cite{Wang2013Performance,gribonval2008atoms} that have focused on analyzing the support recovery guarantee of   noisy  SOMP based on the MIP  when the noise is bounded. With the prior knowledge of support, the work in \cite{gribonval2008atoms} provided a recovery guarantee of  SOMP in the form of $\ell_{p,\infty}$-norm of bounded noise. Additionally, when the noise matrix is fixed, \cite{gribonval2008atoms} derived the successful recovery probability (SRP) with respect to the random matrix $\bC$ in \eqref{obOMP}. Meanwhile, the work in \cite{Wang2013Performance}  analyzed the support recovery guarantee of SOMP which is based on the Frobenius norm of the bounded noise matrix. }

It is worth noting that the existing literature on SOMP lacks an analysis of the SRP with respect to the random and unbounded noise. In some real-world signal applications \cite{OMPchannel,You2022Struc,ZhangMMV2022}, 
the measurement matrix and sparse signal are deterministic, while the  noise is typically random and unbounded. Consequently,  the success of support recovery by SOMP becomes a random event which depends on the distribution of noise. Therefore, it is imperative to investigate  the recovery performance of SOMP when faced with random and unbounded noise.
To address the research gap, we aim to provide a tighter bound
  for the SRP of noisy SOMP  based on the MIP analysis, especially in the scenario when noise is unbounded. This is in contrast with the prior works \cite{Wang2013Performance,gribonval2008atoms} that mainly focused on bounded noise. To achieve a tighter SRP bound, compared to the case of bounded noise, it is essential to explore the specific distribution characteristics of the noise matrix that impact the recovery performance of SOMP. Indeed, this is the challenge of performance analysis in the unbounded noise case.
 By analyzing and evaluating the performance of SOMP,
 we aim to answer several fundamental questions:  (i) what is the SRP of SOMP under an arbitrary distribution of noise? (ii) what is the desired condition of the row-sparse matrix $\bC$ to ensure predefined SRP performance? and (iii) what is the number of sparse vectors (i.e., $d$) to guarantee the exact recovery of the support set when the noise is unbounded?
 To address these questions, we first provide a new condition in terms of the MIP on which the exact support recovery is guaranteed for SOMP for bounded noise. When noise is unbounded, we derive a bound on the SRP that depends on the distribution of the noise matrix.

The contributions of this paper are summarized as follows:
\begin{itemize}
  \item
  We present two iterative SOMP algorithms on the basis of their stopping rule. The first scheme is referred to as SOMP-sparsity (SOMPS), which stops the iteration when the required number of atoms is obtained. The second scheme terminates
  its iteration based on thresholding residual signal power per iteration, which we refer to as SOMP-thresholding (SOMPT).
   When noise power is bounded, i.e., $\| \bN\|_2\le \epsilon$,  we identify the conditions on which the exact support recovery is guaranteed for both SOMPS and SOMPT.
   Unlike the performance guarantee in \cite{Wang2013Performance}, which is based on the Frobenius norm of the noise matrix, we  provide a tighter bound
based on the $\ell_2$-norm of the noise matrix, which leads to new interesting characterizations for SOMP.
  \item
  When noise is unbounded, we derive lower bounds of SRP for both SOMPS and SOMPT algorithms, which depend on a specific distribution of $\ell_2$-norm of the noise matrix.
  Then we focus on a practical scenario when noise is random Gaussian, and derive that the lower bound of SRP follows the Tracy-Widom law distribution \cite{TW2, TW1}.
  From the identified SRP bound, we establish the required noise level, the number of sparse vectors $d$, the number of measurements, and the value of mutual coherence for successful support recovery. We show that the number of
sparse vectors needs to grow proportionally to the noise variance, and the mutual coherence of the measurement
matrix must decrease proportionally to the noise standard
deviation for successful support recovery.

  \item Through numerical simulations, we corroborate the theoretical findings for both SOMPS and SOMPT.
We also illustrate the effect of other factors, such as sparsity $L$, the number of measurements $M$, and the number of sparse vectors $d$, on the recovery performance of noisy SOMPS and SOMPT.
\end{itemize}

\subsection{Paper Organization and Notation}
The paper is organized as follows. In Section \ref{Sproblem formulation}, we introduce the support recovery problem, and present the SOMPS and SOMPT algorithms.
In Section \ref{section preliminary}, some preliminaries and existing results about SOMP are presented.
In Section \ref{section bound noise}, the performances of SOMPS and SOMPT are analyzed for the bounded noise. Then, in Section \ref{section unbounded noise},
we present the performance guarantee of SOMPS and SOMPT when the noise is unbounded, and the case with Gaussian noise is discussed in detail.
The simulation results and conclusions are presented in Section \ref{Ssimulation} and Section \ref{SConclusion}, respectively.

\emph{Notation:}  A bold lower case letter $\mathbf{a}$ is a vector and a bold capital letter $\mathbf{A}$ is a matrix. ${{\mathbf{A}}^{T}}$, ${{\mathbf{A}}^{-1}}$,  ${{\left\| \mathbf{A} \right\|}_{F}}$, $\|  \bA\|_2$ and $\|  \bA\|_{\infty}$ are, respectively, the transpose, inverse, Frobenius norm, $\ell_2$-norm and $\ell_{\infty}$-norm of matrix $\bA$. ${{\left\| \mathbf{a} \right\|}_{1}}$ and ${{\left\| \mathbf{a} \right\|}_{2}}$ are $\ell_1$-norm and $\ell_2$-norm of vector $\mathbf{a}$. $\bA^{\dagger} = (\bA^T \bA)^{-1}\bA^T$ denotes the pseudo inverse of $\bA$ {with linearly independent columns}. ${{[\mathbf{A}]}_{:.i}}$, ${{[\mathbf{A}]}_{i,:}}$, ${{[\mathbf{A}]}_{i,j}}$, and $[\ba]_i$ are, respectively, the $i$th column vector, $i$th row vector, $i$th row and $j$th column entry of $\mathbf{A}$, and $i$th entry of vector $\mathbf{a}$.
${[\mathbf{A}]}_{:,\mathcal{S}}$ denotes a sub-matrix of $\bA$ constructed by taking the columns indexed by set $\mathcal{S}$.
We use $\ba \in \bA$ to
denote a vector $\ba$ chosen from the columns of $\bA$.  $\lambda_{\min}(\bA)$ returns the minimal eigenvalue of $\bA$.
$\R^+$ denotes the set of positive real numbers. $|\mathcal{S}|$ denotes the cardinality of the set $\mathcal{S}$.

\section{Background} \label{Sproblem formulation}

\begin{algorithm} [t]
\caption{SOMPS}
\label{alg_SOMP}
\begin{algorithmic} [1]
\STATE Input: The observations $\bY $, the measurement matrix  $\bPhi$, sparsity level $L$.
\STATE Initialization: Support set $\widehat{\Omega}^{(0)}= \emptyset$, residual matrix $\bR^{(0)} = \bY$.
\FOR{$l = 1$ to $L$}
\STATE Select the largest index $\eta = \argmax \limits _{i=1,\cdots,N} \| [\bPhi]_{:,i}^T \bR^{(l-1)} \|_2$. \label{select rule}
{\STATE Update the support set: $\widehat{\Omega}^{( l)}= \widehat{\Omega}^{( l-1)} \bigcup \{\eta\}$. \label{update rule}}
\STATE Update the residual matrix: $\bR^{(l)} =\bY- [{\bPhi}]_{:,\widehat{\Omega}^{( l)} } [{\bPhi}]_{:,\widehat{\Omega}^{( l)} }^{\dagger}\bY  $. \label{alg update}
\ENDFOR
\STATE Output: $\widehat{\Omega}^{( L)}$.
\end{algorithmic}
\end{algorithm}

This section introduces the SOMPS and SOMPT algorithms and outlines our motivation to utilize the mutual coherence for analyzing these algorithms.

\subsection{SOMP Algorithms} \label{sect:SOMP}
 Given the MMV observation model in  \eqref{obOMP},
we recall that only $L$ rows of $\bC$ are non-zero. Here,
we denote the row support set of $\bC$  as $\Omega \subset \{1,2,\ldots,N  \}$ with $|\Omega|=L$.
To estimate the $\Omega$ from the observations $\bY=\bPhi \bC+\bN$ in \eqref{obOMP}, in this work, we mainly focus on two specific SOMP algorithms \cite{ChenMMV,Tropp2006Greedy,SOMP_l2,SOMP_l2_T}, where  Algorithm \ref{alg_SOMP} corresponds to SOMPS and Algorithm \ref{alg_SOMP_stop} corresponds to SOMPT.
When the row sparsity of $\bC$ (i.e., $L$) is known as a priori, the iteration of SOMPS is terminated when the required number of atoms are selected as described in Algorithm \ref{alg_SOMP}.
On the other hand, when the row sparsity of $\bC$ is unavailable as a priori,
 a threshold $\tau$ can be introduced to evaluate against the amount of power in the residual matrix $\bR^{(l)}$ at each iteration $l$.
Specifically, when $\|\bR^{(l)}\|_{2}$ is less than the threshold $\tau$ in Step \ref{step tau} of Algorithm \ref{alg_SOMP_stop}, the iteration of SOMPT is terminated.
In summary, when the row sparsity of the signal is known a priori, SOMPS can be applied, ensuring that the output of the algorithm strictly adheres to the required sparsity.
Conversely, when the row sparsity of the signal is unknown, SOMPT is applicable. Thus, SOMPT can be used in a broader range of scenarios where the exact sparsity level is uncertain or unavailable. 
Regarding complexity, since these two algorithms are similar except for the stopping criterion, the computational complexity for SOMPS and SOMPT are both $\mathcal{O}(dLMN)$.

As for each iteration of SOMPS and SOMPT, the active index determined in Step \ref{select rule}, i.e.,
\begin{align}
\eta = \argmax_{i=1,\ldots,N} \| [\bPhi^T \bR^{(l-1)}]_{i,:} \|_2, \nonumber
\end{align}
is added to the previously detected support set $\widehat \Omega^{(l-1)} $ to form $\widehat \Omega^{(l)}$ in Step \ref{update rule}.
It is crucial to recognize that the updated residue $\bR^{(l)}$ in Step \ref{alg update} of Algorithm \ref{alg_SOMP} and Algorithm \ref{alg_SOMP_stop} is orthogonal to the columns of $[{\bPhi}]_{:,\widehat{\Omega}^{( l)} }$. Denoting $\bP^{(l)} =[\bPhi]_{:,\widehat{\Omega}^{(l)} }  [\bPhi]_{:,\widehat{\Omega}^{(l)} }^{\dagger}\in \R^{M \times M}$ and $\bP^{(l)}_{\perp}=\mathbf{I}-{{\mathbf{P}}^{(l)}}$, the residual of $l\text{th}$ iteration is expressed as
\begin{align}
{{\mathbf{R}}^{(l)}}
=\bP^{(l)}_{\perp}\bY=\bP^{(l)}_{\perp}({\bPhi}{\bC}+{{\mathbf{N}}}),\label{RL}
\end{align}
where the columns of residue ${{\mathbf{R}}^{(l)}}$ belong to the column subspace of $\bP^{(l)}_{\perp}$. It is for this reason that it is called orthogonal matching pursuit.

It is worth noting that Algorithm \ref{alg_SOMP} and Algorithm \ref{alg_SOMP_stop} successfully recover the support if and only if each active index determined in Step \ref{select rule} is in the support set, i.e., $\eta \in \Omega$.
In particular, given that the first $l$ iterations of Algorithm \ref{alg_SOMP} and  Algorithm \ref{alg_SOMP_stop} selected $l$ atoms correctly, the following remark gives a condition for selecting the correct atom at the $(l+1)$th iteration.

Assume that the first $l$ iterations of SOMPS in Algorithm \ref{alg_SOMP} and SOMPT in Algorithm \ref{alg_SOMP_stop} selected $l$ correct atoms, i.e., $\widehat{\Omega}^{(l)}\subset \Omega$, and the termination condition is not satisfied at the $(l+1)$th iteration\footnote{The fact that the termination condition is not satisfied at the $(l+1)$th iteration means that $l+1\le L$ for Algorithm \ref{alg_SOMP}, and $\| \bR^{(l)}\|_2 \ge \tau$ for Algorithm \ref{alg_SOMP_stop}.}.
The $(l+1)$th iteration will select the correct atom when  the following holds \cite{Tropp2006Greedy,OMPTT},
\begin{align}
\underset{\mathbf{d}\in [\bPhi]_{:,{\Omega}}}{\mathop{\max }}\,{{\left\| {{\mathbf{d}}^T}{{\mathbf{R}}^{(l)}} \right\|}_{2}}>\underset{\mathbf{d}\in [\bPhi]_{:,{\Omega^c}}}{\mathop{\max }}\,{{\left\| {{\mathbf{d}}^T}{{\mathbf{R}}^{(l)}} \right\|}_{2}}, \label{guarantee con}
\end{align}
where $\Omega^c \subset \{ 1,\ldots ,N\}$ with $|\Omega^c|=N-L$ denotes the {complement} of the set $\Omega$.
The condition in \eqref{guarantee con} is the key to the analysis for the performance guarantee of SOMP.

\begin{algorithm} [t]
\caption{SOMPT}
\label{alg_SOMP_stop}
\begin{algorithmic} [1]
\STATE Input: The observations $\bY $, the measurement matrix  $\bPhi$, the threshold $\tau$.
\STATE Initialization: Support set $\widehat{\Omega}^{(0)}= \emptyset$, residual matrix $\bR^{(0)} = \bY$, iteration number $l=1$.
\WHILE{$\| \bR^{(l-1)}\|_2  \ge  \tau$} \label{step tau}
\STATE Select the largest index $\eta = \argmax \limits _{i=1,\ldots,N} \| [\bPhi]_{:,i}^T \bR^{(l-1)} \|_2$.
{\STATE Update the support set: $\widehat{\Omega}^{( l)}= \widehat{\Omega}^{( l-1)} \bigcup \{\eta\}$. \label{residue of iteration}}
\STATE Update the residual matrix: $\bR^{(l)} =\bY- [{\bPhi}]_{:,\widehat{\Omega}^{( l)} } [{\bPhi}]_{:,\widehat{\Omega}^{( l)} }^{\dagger}\bY  $, and $l  \leftarrow l+1$.
\ENDWHILE
\STATE Output: $\widehat{\Omega}^{(l-1)}$.
\end{algorithmic}
\end{algorithm}

\subsection{Motivations} \label{section motivation}
{
The mutual coherence in \eqref{def of mu} has been utilized to measure the maximal coherence of different columns of the measurement matrix \cite{Donoho2003MIP,DonohoMIP}.
Quantifying the mutual coherence of a matrix is crucial in analyzing and solving various signal processing problems, including the Grassmannian line packing \cite{str2003Grass, love2003Grass}, Riemannian manifold packing \cite{Kim11-1}, support detection \cite{TroppOMP, OMPTT,duanOMP},
and the evaluation of the focusing capabilities of imaging systems \cite{Obermeier2017Sensing}. As for the support recovery, the mutual coherence is also crucial to quantify the guarantee of the successful recovery \cite{OMPTT,TroppGreedy, MuOMP,Mi2017Prob }. For example, it has been shown that OMP can successfully detect the support set in the noiseless scenario when $\mu < {1}/{(2L-1)}$ \cite{MuOMP,TroppGreedy}.
Alternatively, the RIP of the measurement matrix is also an important characteristic for support recovery. When the RIC satisfies $\delta_{L+1} < 1/\sqrt{L+1}$ \cite{sharpOMP}, the exact support recovery is guaranteed for OMP in the noiseless case.
While the mutual coherence value of the measurement matrix can induce the RIC, e.g., $\delta_{L} \le (L-1) \mu$ \cite{foucart2013mathematical,Cai09On}, it is generally challenging to calculate the RIC of a given measurement matrix. Unlike the RIC, the mutual coherence can be calculated for a given measurement matrix.
Moreover, the mutual coherence of the measurement matrix $\bPhi\in\R^{M\times N}$ can be calculated, based on \eqref{def of mu}, and satisfies the Welch bound
 \cite{Welch1974bound, DonohoMIP},
\begin{align}
 \mu \ge \sqrt{\frac{N-M}{M(N-1)}}. \label{bound mu 1}
\end{align}
The synthesis of the measurement matrix that nearly achieves the bound in \eqref{bound mu 1} can be found by using the methodologies in \cite{str2003Grass,Herman2009High,elad_CS}.

While it is true that the MIP reveals its amenability and has been exploited in various signal processing problems, thoroughly understanding the noisy SOMP in terms of the MIP is still nascent.}
To this end, our goal in this work is to provide the guarantee conditions for successful support recovery of Algorithm \ref{alg_SOMP} and Algorithm \ref{alg_SOMP_stop} in terms of
MIP for both bounded and unbounded noise.

\section{Preliminaries} \label{section preliminary}
Before embarking on the guarantee analysis of SOMP, we summarize several results presented in \cite{TroppGreedy,Wang2013Performance,OMPTT}, and derive useful conditions that we will rely on.

\begin{Lemma}[Property of matrix $\bPhi${\cite{TroppGreedy}}]   \label{bound G}
For the model in \eqref{obOMP}, define the constant
\begin{align}
G = \max_{\ba \in  [\bPhi]_{:,\Omega^c}} \lA([\bPhi]_{:,\Omega}^T[\bPhi]_{:,\Omega})^{-1} [\bPhi]_{:,\Omega}^T \ba \rA_1. \label{def G}
\end{align}
Then, the value of $G$ is upper bounded by $G\le \frac{L\mu}{1-(L-1)\mu}$,  where the constant $\mu$ is the mutual coherence of matrix $\bPhi$.
\end{Lemma}

\begin{Lemma}[Minimal eigenvalue inequality \cite{OMPTT}] \label{eigen value bound}
The minimum eigenvalue of $[\bPhi]^T_{:,{\Omega}}  [\bPhi]_{:,{\Omega}}$ is less than or equal to the minimum eigenvalue of $[\bPhi]^T_{:,\widehat{\Omega}_{c}^{(l)}} \bP^{(l)}_{\perp} [\bPhi]_{:,\widehat{\Omega}_{c}^{(l)}}$, where $\widehat{\Omega}_c^{(l)}$ denotes the complement of the selected support set   $\widehat{\Omega}^{(l)}$ over the universe $\Omega$ with $|\widehat{\Omega}_c^{(l)}| = L-l$.
\end{Lemma}

Though the condition in \eqref{guarantee con} is necessary and sufficient for the correct selection of the $(l+1)$th atom,  it is not practical to check whether the inequality in \eqref{guarantee con} holds. Because it {depends on prior information} of support $\Omega$.
An alternative condition for \eqref{guarantee con} is derived by plugging the expression of $\bR^{(l)}$ in \eqref{RL} into \eqref{guarantee con} and defining
the following quantities \cite{Wang2013Performance,OMPTT},
\begin{align}
{{Q}^{(l,1)}}&=\underset{\mathbf{d}\in [\bPhi]_{:,\Omega} }{\mathop{\max }}\,{{\left\| {{\mathbf{d}}^T}\bP^{(l)}_{\perp}\bPhi\bC \right\|}_{2}},\label{M1} \\
 {{Q}^{(l,2)}}&=\underset{\mathbf{d}\in [\bPhi]_{:,\Omega^c}}{\mathop{\max }}\,{{\left\| {{\mathbf{d}}^T}\bP^{(l)}_{\perp}\bPhi\bC \right\|}_2} ,\label{M2}\\
 {{Z}^{(l)}}&=\underset{\mathbf{d}\in \bPhi}{\mathop{\max }}\,{{\left\| {{\mathbf{d}}^T}\bP^{(l)}_{\perp} \mathbf{N} \right\|}_{2}}. \label{ZL}
 \end{align}
Then, it has the following lemma.

\begin{Lemma}[\! \cite{Wang2013Performance,OMPTT}] \label{relation M1 M2}
 Let ${{Q}^{(l,2)}} $ and ${{Q}^{(l,1)}}$ be defined in \eqref{M1} and \eqref{M2}, respectively. Then the following inequality holds,
\begin{align}
{{Q}^{(l,2)}} \le G{{Q}^{(l,1)}}, \label{ineq M12}
\end{align}
where $G$ is defined in \eqref{def G}.
\end{Lemma}

Using the notations in \eqref{M1}, \eqref{M2}, and \eqref{ZL}, the following lemma expresses the condition in \eqref{guarantee con} in terms of ${{Q}^{(l,1)}}$ and $ {{Z}^{(l)}}$.
{
\begin{Lemma}[Sufficient condition for correct selection at the $(l+1)$th iteration
\cite{Wang2013Performance}] \label{pre condition 1}
Suppose the definitions in  \eqref{M1},  \eqref{M2},  and \eqref{ZL}. If the first $l$ iterations select the correct atoms, the sufficient condition for selecting the correct $(l+1)$th atom is
\begin{align}
{{Q}^{(l,1)}}-{{Q}^{(l,2)}}>2{{Z}^{(l)}}.\label{sufficient condition 1}
\end{align}
\end{Lemma}}

{
\begin{Proposition} \label{propositionOne}
Combining \eqref{ineq M12}  with \eqref{sufficient condition 1}, a condition for correct selection at the $(l+1)$th iteration can be given by
\begin{align}
 (1-G){{Q}^{(l,1)}} >2{{Z}^{(l)}}. \label{sufficient condition 1 temp 1}
\end{align}
Furthermore, if $\mu < 1/(2L-1)$ holds to guarantee $1 - G >0$, plugging the bound of $G$ in Lemma \ref{bound G} into  \eqref{sufficient condition 1 temp 1} leads to
\begin{align}
{{Q}^{(l,1)}}>2\frac{1-(L-1)\mu }{1-(2L-1)\mu }{{Z}^{(l)}}. \label{sufficient condition 1 temp}
\end{align}
Thus, when the condition in \eqref{sufficient condition 1 temp} holds, the $(l+1)$th iteration selects the correct atom.
\end{Proposition}
}
It is worth noting that the condition $\mu < 1/ \left( {2L - 1} \right)$  is imposed to guarantee  $1-G>0$ in \cref{sufficient condition 1 temp 1}, which is obtained by substituting $G\le \frac{L\mu}{1-(L-1)\mu}$ in \cref{bound G}. Otherwise, if $\mu \ge 1/ \left( {2L - 1} \right)$, 
the condition in \cref{sufficient condition 1 temp 1} may be violated since $Z^{(l)} \ge 0$ , which leads to the incorrect selection of the $(l+1)$th atom. Moreover, one can find that there exists  a measurement matrix $\bPhi$ with mutual coherence $\mu=1/ \left( {2L - 1}\right)$ and  $L$-sparse signal $\bC$ such that $(l+1)$th atom of $\bC$ cannot be correctly selected by SOMP.
Note that the condition $\mu < 1/ \left( {2L - 1} \right)$  is also a necessary condition  for the successful recovery of sparse signal by using OMP \cite{MuOMP,Tropp2006Greedy}. Since SOMP is a general extension of OMP, 
this further confirms the consistency of our analysis with that of OMP \cite{MuOMP}.

Observing \eqref{sufficient condition 1 temp}, computing $Q^{(l,1)}$ in \eqref{M1} to check whether the condition in \eqref{sufficient condition 1 temp} holds still relies on a priori knowledge $\Omega$.
The following lemma provides a bound of $Q^{(l,1)}$ that only depends on the non-zero rows of $\bC$, i.e., does not depend on the a priori knowledge $\Omega$.

\begin{Lemma}[Bound of ${{Q}^{(l,1)}}$ \cite{Wang2013Performance}] \label{bound M1}
The ${{Q}^{(l,1)}}$ defined in \eqref{M1} is lower bounded by
\begin{align}
   {{Q}^{(l,1)}} \ge {{(L-l)}^{-1/2}}(1-(L-1)\mu ){{\lA [\bC]_{\widehat{\Omega}_c^{(l)},:}\rA}_F}.
\end{align}
\end{Lemma}

In summary, according to Lemma \ref{pre condition 1} and Lemma \ref{bound M1}, in order to guarantee the $(l+1)$th iteration selects the correct atom, the sufficient condition in \eqref{sufficient condition 1 temp} can be rewritten as,
\begin{align}
\frac{1-(L-1)\mu}{(L-l)^{1/2}}{{\left\| [\bC]_{\widehat{\Omega}_c^{(l)},:} \right\|}_{F}}>2\frac{1-(L-1)\mu }{1-(2L-1)\mu }{{Z}^{(l)}} ,	\nonumber
\end{align}
which is simplified to
\begin{align}
	{{\left\| [\bC]_{\widehat{\Omega}_c^{(l)},:} \right\|}_{F}}&>\frac{2\sqrt{L-l}}{1-(2L-1)\mu } {{Z}^{(l)}}. \label{Q values}
\end{align}

\section{Guarantee of Recovery Under Bounded Noise} \label{section bound noise}
In this section, we present the recovery guarantee of SOMP when the noise is upper bounded, i.e., $\| \bN \|_2 \le \epsilon$.

\subsection{SOMPS Under Bounded Noise}
First of all, we extend the analysis for OMP in \cite{OMPTT} to the MMV model in \eqref{obOMP} and establish a condition that guarantees the successful recovery of the support set by using SOMPS in Algorithm \ref{alg_SOMP}.

 \begin{Theorem} \label{theorem SOMP}
 Given the signal model in \eqref{obOMP}  and the mutual coherence $\mu$ of the measurement matrix $\bPhi$ satisfying $\mu < 1/(2L-1)$, if
\begin{align}
C_{\min} > \frac{2 \| \bN \|_2}{{1-(2L-1)\mu }}, \label{theorem alg1}
\end{align}
where $C_{\min}= \min\limits _{i\in \Omega}  \lA [\bC]_{i,:} \rA_2$, then SOMPS in Algorithm \ref{alg_SOMP} successfully recovers the support set $\Omega$.
In particular, when the noise
matrix is bounded $\| \bN\|_2 \le \epsilon$, the condition in \eqref{theorem alg1} becomes
\begin{align}
C_{\min} > \frac{2 \epsilon}{{1-(2L-1)\mu }}.\label{theorem bounded noise alg1}
\end{align}
 \end{Theorem}
 \begin{proof}
 First off, a sufficient condition that guarantees the inequality in \eqref{Q values} can be given by
\begin{align}
{{\left\| {{[\bC]}_{i,:}} \right\|}_{2}}> \frac{2{{Z}^{(l)}}}{1-(2L-1)\mu }, \forall i \in \Omega. \label{C with L}
\end{align}
 Moreover, it is noted from \eqref{ZL} that
\begin{align}
{{Z}^{(l)}}=\underset{\mathbf{d}\in \bPhi}{\mathop{\max }}\,{{\| {{\mathbf{d}}^T}\bP^{(l)}_{\perp}\mathbf{N} \|}_{2}}
\le \underset{\| \mathbf{d}\|_2 = 1}{\mathop{\max }}\,{{\left\| \bd^T {\mathbf{N} }\right\|}_{2}}
= \lA{\bN}\rA_2.\label{bound for Z l}
\end{align}
Combining \eqref{C with L} and \eqref{bound for Z l}, if
\begin{align}
{{\left\| {{[\bC]}_{i,:}} \right\|}_{2}}> \frac{2 \| \bN\|_2 }{1-(2L-1)\mu}, \forall i\in \Omega, \nonumber
\end{align}
and the  Algorithm \ref{alg_SOMP} terminates when $l=L$, SOMPS
successfully recovers the support $\Omega$. This concludes the proof.
 \end{proof}


According to Theorem \ref{theorem SOMP}, the exact recovery of the support set is guaranteed for Algorithm \ref{alg_SOMP} if $C_{\min}$ is lower bounded by the right-hand side of \eqref{theorem alg1}. The derived lower bound is dependent on the noise level $\| \bN\|_2$, sparsity level $L$, and the mutual coherence $\mu$. When the values of $L$ and $\mu$ are fixed, in order to guarantee the successful recovery of the support set, the value of $C_{\text{min}}$ should be proportional to the value of $\| \bN\|_2$.
Note that this reveals a tighter bound than the Frobenius-norm bound in \cite{Wang2013Performance} because of the fact that $\| \bN\|_2 \le \| \bN \|_F$. {With this tighter bound, we can obtain a tighter SRP guarantee when the noise is unbounded in Section \ref{section unbounded noise}.}
In particular,
Theorem \ref{theorem SOMP} is a generalization of OMP in \cite{OMPTT}, where the number of sparse vectors $d=1$.

\begin{Remark} \label{condition of mu somps}
By Theorem \ref{theorem SOMP}, we can obtain the condition of  $\mu$ to guarantee the successful recovery of SOMPS as follows
\begin{align}
\mu < \frac{1-  {2\| \bN \|_2}/{C_{\text{min}}}}{2L-1}. \label{condition mu}
\end{align}
We note that the right-hand side of \eqref{condition mu} is smaller
 than the noiseless case, i.e., $\mu < 1/(2L-1)$ \cite{TroppGreedy,MuOMP}, which is due to the noisy measurements.
When the noise is bounded, i.e., $\| \bN \|_2\le \epsilon$, the condition in \eqref{condition mu} is given by $\mu < \frac{1-  {2\epsilon}/{C_{\text{min}}}}{2L-1}.$
\end{Remark}

\begin{Remark} \label{condition of M}
If \eqref{condition mu} is combined with the Welch bound in \eqref{bound mu 1}, we obtain a condition for the performance guarantee of SOMPS when measurement matrix $\bPhi$ achieves the minimal mutual coherence. Specifically, when the following holds
\begin{align} \label{bound wel M}
\sqrt{\frac{N-M}{M(N-1)}} < \frac{1-  {2\epsilon}/{C_{\text{min}}}}{2L-1},
\end{align}
the successful recovery of SOMPS is guaranteed.
\end{Remark}


The effect of the number of measurements $M$ on the performance of SOMPS is of interest.
{Considering $\bN \in \R^{M\times d}$,  the bound of $\| \bN\|_2\le \epsilon$ is a function of the number of measurements $M$.  However, due to the correlation between $M$ and $\epsilon$, it is difficult
to state generally a relationship between $M$ and the recovery performance of SOMPS from \eqref{bound wel M}.}
 Instead, we focus on the case when $\epsilon$ is a constant and $1-  {2\epsilon}/{C_{\text{min}}}>0$. Rearranging the inequality in \eqref{bound wel M} leads to the condition for the required number of $M$ as follows
\begin{align}
M > \frac{N}{\left(\frac{1-  {2\epsilon}/{C_{\text{min}}}}{2L-1}\right)^2 (N-1)+1}. \label{bounded for M}
\end{align}
It means that when the number of measurements $M$ is larger than the right hand side of \eqref{bounded for M}, the successful support recovery can be guaranteed for all sparse signals $\bC$ with sparsity $L$ and $C_{\min}$. In particular, when $N\gg L$ and noise level is low $ {\epsilon}/{C_{\text{min}}} \approx 0$, the condition in \eqref{bounded for M} can be approximated to $M > (2L-1)^2$.
Compared to the well-known bound on the required number of measurements $M=O(L\log(N))$ in \cite{TroppOMP,Candes2006Robust}, the derived bound $M > (2L-1)^2$ from \eqref{bounded for M} is indeed independent of $N$ and becomes tighter as the matrix $\bC$ becomes more sparse, i.e., as $L $ decreases.\footnote{Note that the condition in \eqref{bounded for M} is derived in a different manner from the analysis of phase transition in \cite{TroppOMP,Candes2006Robust}. In our work, the measurement matrix is deterministic, and the guarantee condition is for all sparse signals $\bC$ with the sparsity $L$  and required $C_{\min}$.
The measurement matrix is random in \cite{TroppOMP,Candes2006Robust}, and the condition for the number of measurements is in a statistical sense.}

\subsection{SOMPT Under Bounded Noise}
Unlike SOMPS in Algorithm \ref{alg_SOMP}, the stopping criterion of SOMPT is determined by the threshold value $\tau$ in Algorithm \ref{alg_SOMP_stop}.
The following theorem provides a sufficient condition for the successful support recovery of SOMPT.
 \begin{Theorem} \label{theorem SOMP stop}
 Given the signal model in \eqref{obOMP}, suppose that the mutual coherence $\mu$ of the measurement matrix $\bPhi$ satisfies $\mu < 1/(2L-1)$, the noise is bounded by $\|\bN\|_2 \leq \epsilon$, and the threshold value  is $\tau = \epsilon$ for SOMPT in Algorithm \ref{alg_SOMP_stop}. Then if the following condition holds,
\begin{align}
C_{\min} > \frac{2 \epsilon}{{1-(2L-1)\mu }},  \label{assump theorem 2}
\end{align}
 Algorithm \ref{alg_SOMP_stop} successfully recovers the support set $\Omega$.
 \end{Theorem}
\begin{proof}
By Theorem \ref{theorem SOMP}, we only need to prove that the Algorithm \ref{alg_SOMP_stop} terminates correctly, i.e., $l=L$.
In other words, it is sufficient to show that when $l < L$ we have $\| \bR^{(l)}\|_2 > \epsilon$, and when $l  =  L$ we have $\| \bR^{(l)}\|_2 \le \epsilon$.
First of all, when $l  =  L$, the following holds,
\begin{align} \label{proof of RL}
\| \bR^{(L)}\|_2  &= \| \bP^{(L)}_{\perp}\bPhi\bC + \bP^{(L)}_{\perp}\bN \|_2 \nonumber\\
&=\|  \bP^{(L)}_{\perp} \bN \|_2 \nonumber \\
&\le  \| \bN \|_2
\le  \epsilon.
\end{align}
When $l < L$,  one can have the following,
\begin{align}
\| \bR^{(l)}\|_2&=  \| \bP^{(l)}_{\perp}\bPhi\bC + \bP^{(l)}_{\perp}\bN \|_2 \nonumber\\
&\ge   \| \bP^{(l)}_{\perp}\bPhi\bC\|_2  -\| \bP^{(l)}_{\perp}\bN \|_2  \nonumber \\
&\ge  \| \bP^{(l)}_{\perp}\bPhi\bC\|_2 -\epsilon \nonumber\\
&\overset{(a)}{=} \| \bP^{(l)}_{\perp} [\bPhi]_{:,\widehat{\Omega}_{c}^{(l)}}  [\bC]_{\widehat{\Omega}_c^{(l)},:} \|_2- \epsilon \nonumber \\
&\overset{(b)}{\ge} \sqrt{\lambda_{\text{min}}\left([\bPhi]^T_{:,\widehat{\Omega}_{c}^{(l)}} \bP^{(l)}_{\perp} [\bPhi]_{:,\widehat{\Omega}_{c}^{(l)}} \right) }\| [\bC]_{\widehat{\Omega}_c^{(l)},:} \|_2- \epsilon \nonumber\\
&\overset{(c)}{\ge} \sqrt{\lambda_{\text{min}}\left([\bPhi]^T_{:,{\Omega}}  [\bPhi]_{:,{\Omega}} \right) }\| [\bC]_{\widehat{\Omega}_c^{(l)},:} \|_2- \epsilon \nonumber\\
&\overset{(d)}{\ge}  \sqrt{1-(L-1)\mu}\| [\bC]_{\widehat{\Omega}_c^{(l)},:} \|_2- \epsilon, \label{R stop 1}
\end{align}
{where $(a)$ is due to the definition of projection matrix $\bP^{(l)}_{\perp}$, i.e., $\bP^{(l)}_{\perp} =\bI- [\bPhi]_{:,\widehat{\Omega}^{(l)} }  [\bPhi]_{:,\widehat{\Omega}^{(l)} }^{\dagger}$, $(b)$ follows from the idempotency and symmetry properties of orthogonal projection matrices, $(c)$ is due to Lemma \ref{eigen value bound}, and $(d)$ is due to the Gershgorin disk theorem \cite{gershgorin1931uber}}.

{
Given the condition in \eqref{assump theorem 2}, the first term on the right hand side of \eqref{R stop 1} can be further lower bounded by
\begin{align}
\sqrt{1-(L-1)\mu}\| [\bC]_{\widehat{\Omega}_c^{(l)},:} \|_2
&\ge \sqrt{1-(L-1)\mu} \min_{i \in \Omega}\| [\bC]_{i,:}\|_2 \nonumber \\
 &\overset{(a)}{>} \frac{2 \sqrt{1-(L-1)\mu}\epsilon }{1-(2L-1)\mu }\nonumber \\
  &\overset{(b)}{\ge}2\epsilon, \label{R stop 2}
\end{align}
where $(a)$ is due to \eqref{assump theorem 2} and $(b)$ follows from the fact that $\sqrt{1-(L-1)\mu} \ge 1-(2L-1)\mu$. Combining \eqref{R stop 1} with \eqref{R stop 2} leads to $\| \bR^{(l)}\|_2 > \epsilon, \forall l <L $. This concludes the proof.}
\end{proof}

Based on the results of Theorem \ref{theorem SOMP stop}, if we further assume that the measurement matrix $\bPhi$ achieves the Welch bound in \eqref{bound mu 1}, the number of measurements satisfies
\begin{align}
M > \frac{N}{\left(\frac{1-  {2 \epsilon}/{C_{\text{min}}}}{2L-1}\right)^2 (N-1)+1}, \nonumber
\end{align}
and the threshold is set to $\tau = \epsilon$ in Algorithm \ref{alg_SOMP_stop}, SOMPT successfully recovers the support set $\Omega$.

\section{Guarantee of Recovery Under Unbounded Random Noise} \label{section unbounded noise}
Unlike the previous section, we establish in this section the recovery guarantee of SOMPS and SOMPT when the noise is unbounded.

\subsection{Analysis of SOMPS Under Unbounded Noise} \label{sec:SOMPS}

\subsubsection{SRP of SOMPS under General Random Noise}

The following theorem quantifies the successful recovery probability (SRP) of SOMPS in Algorithm \ref{alg_SOMP} under general random noise.
{
\begin{Theorem} \label{theorem SOMP random L}
Given the model in \eqref{obOMP}, suppose that mutual coherence $\mu$ of the measurement matrix $\bPhi$ satisfies $\mu< 1/(2L-1)$.  If the CDF $F_{N}(\cdot)$  of $\| \bN\|_2$ is defined by
\begin{align} \label{expression of FN}
\Pr(\|  \bN \|_2 \le x) =  F_{N}(x), \forall x >0,
\end{align}
the SRP  $P_s$ of SOMPS in Algorithm \ref{alg_SOMP} is lower bounded by
 \begin{align}
   P_s \ge F_{N}\left(\frac{C_{\text{min}}{{(1-(2L-1)\mu )}} } {2 }\right), \label{SOMPS SRP}
  \end{align}
  where $C_{\min} = \min\limits _{i\in \Omega}  \lA [\bC]_{i,:} \rA_2$.
\end{Theorem}
\begin{proof}
When $\bN$ is random,  $\| \bN\|_2$ in \eqref{theorem alg1} is also a random variable. Thus, by Theorem \ref{theorem SOMP}, when the condition in \eqref{theorem alg1} holds, SOMPS in Algorithm \ref{alg_SOMP} guarantees the successful recovery of support. To describe such an event, we define  the following two events for an arbitrary $x>0$,
\begin{align}
\mathcal{X}_1 &= \left\{\lA [\bC]_{i,:} \rA_2 > \frac{2 x }{{1-(2L-1)\mu }}, \forall i \in \Omega \right\}, \nonumber \\
\mathcal{X}_2 &= \left\{ \| \bN\|_2 \le x \right\} \nonumber.
\end{align}
Then, according to Theorem \ref{theorem SOMP}, the SRP of SOMPS is lower bounded by $\text{Pr}(\mathcal{X}_1\cap \mathcal{X}_2 )$. Moreover, since $x$ is arbitrary, we let $x={C_{\min}^-  {(1-(2L-1)\mu )}}/{2}$, where $C_{\min}^-$ denotes the left-sided limit of $C_{\min}$. Then
{\begin{align*}
\text{Pr}(\mathcal{X}_1\cap \mathcal{X}_2 )&=\text{Pr}( \mathcal{X}_2 )\\
&=F_{N}\left(\frac{C_{\min}^-  {(1-(2L-1)\mu )}}{2}\right)\\
&\overset{(a)}{=}F_{N}\left(\frac{C_{\min}  {(1-(2L-1)\mu )}}{2}\right),
\end{align*}
where the equality $(a)$ follows from the fact that $F_{N}(\cdot)$ is continuous.} Thus, the SRP of Algorithm \ref{alg_SOMP} satisfies
 \begin{align}
   P_s \ge F_{N}\left(\frac{C_{\text{min}}{(1-(2L-1)\mu )} } {2 }\right).
 \nonumber
  \end{align}
 This concludes the proof.
\end{proof}
}
Compared to the bounded noise case in Theorem \ref{theorem SOMP}, the successful support recovery with unbounded noise is described in terms of SRP.
Seen from Theorem \ref{theorem SOMP random L}, the lower bound of  SRP is determined by the distribution of $\|  \bN\|_2$, i.e., $F_{N}(\cdot)$.

\subsubsection{SRP of SOMPS under Gaussian Noise}
One can obtain the expression of  $F_{N}(\cdot)$ in \eqref{expression of FN} for a given distribution of $\bN$.
In the following, we show the case when the entries in $\bN$ are independent and identically distributed (i.i.d.) Gaussian $\cN(0,\sigma^2)$.
\begin{Proposition} \label{gaussian l2}
Suppose that random matrix $\bN \in \R^{M \times d}$ has entries i.i.d. according to $\cN(0, \sigma^2)$. The CDF of the largest singular value of $\bN$ converges in distribution to the Tracy-Widom law  \cite{TW2, TW1}
 as $M \rightarrow \infty$ and  $d/M \rightarrow \gamma  \in (0,\infty)$,
\begin{align}
\Pr\left(\frac{\| \bN\|_2}{\sigma}\le \sqrt{s\sigma_{M,d} + \mu_{M,d}}\right) \overset{d}{\rightarrow} F_1\left(s \right), \label{guanssian dis}
\end{align}
where
\begin{align*}
 \mu_{M,d}  &=\left(\sqrt{M-\frac{1}{2}}  + \sqrt{d-\frac{1}{2}}\right) ^2, \\
\sigma_{M,d}& =\left (\sqrt{M-\frac{1}{2}}  + \sqrt{d-\frac{1}{2}}\right)\left(\frac{1}{\sqrt{M-\frac{1}{2}}}  + \frac{1}{\sqrt{d-\frac{1}{2}}}  \right) ^{1/3},
\end{align*}
and  the function $F_1(\cdot)$ is the CDF of Tracy-Widom distribution \cite{TW1,TW2}, which is expressed as
\begin{align}
  F_1(s) = \exp\left(-\frac{1}{2} \int_s^{\infty}q(x)+(x-s)q(x) dx\right), \nonumber
  \end{align}
  where $q(x)$ is the solution of Painlev\'{e} equation of type II:
  \begin{align}
  q''(x)=xq(x)+2q(x)^3,~ q(x) \sim \text{Ai}(x), x \rightarrow \infty, \nonumber
  \end{align}
  where $\text{Ai}(x)$ is the Airy function \cite{TW2, TW1}.
  \end{Proposition}
\begin{proof}
Proposition \ref{gaussian l2} directly follows the derivations provided in \cite{TW2, TW1}. For the  completeness of the paper,  we provide a sketch of the proof in the following.

For a single Wishart matrix $\bA=\bX \bX^T $ with $ \bX \in \R^{p \times n}$ having i.i.d. Gaussian entries, i.e., $\cN(0,1)$, the joint distribution of the eigenvalues of $\bA$ is given by
\begin{align}
f(x_1,\ldots, x_p) = c\prod_{i} w^{1/2}(x_i) \prod_{i < j} (x_i - x_j) , x_1\ge  \ldots \ge x_p, \nonumber
\end{align}
where $w(x)=x^{n-p-1}e^{-x}$ and $c$ is a normalization constant. Then, the CDF of largest eigenvalue $x_{1}$, i.e., $\Pr(x_1 \le t)$,
can be expressed using the concept of Fredholm determinants, as described in
 in \cite{TW2, TW1,johnstone2007high}. By employing asymptotic analysis techniques, we can obtain the final result of distribution of the largest eigenvalue $x_1$ of the Wishart matrix $\mathbf{A}$ in \cref{guanssian dis}.
\end{proof}

To save the computational complexity, we adopt the table lookup method \cite{dataTW} to obtain the value of $F_1(\cdot)$.

\begin{Remark}[Second order accuracy of Tracy-Widom \cite{johnstone2007high}] \label{second order accuracy} It is noted that the convergence in \eqref{guanssian dis} is in distribution, and the accuracy of the convergence result is characterized by
\begin{align}
\left| \Pr\left(\! \frac{\| \bN\|_2}{\sigma}\! \le \!  \sqrt{s\sigma_{M,d} + \mu_{M,d}}\! \right)\!\! -\!\! F_1\left(s \right)\right|
 \!  \le  \! C \exp(-cs)d^{-2/3}, \label{second order bound}
\end{align}
where $C,c$ are some constants.
For convenience, we substitute ${\sigma} \sqrt{s\sigma_{M,d}+ \mu_{M,d}}= x$ in \eqref{second order bound}. Then, the relation in \eqref{second order bound} can be equivalently written as
\begin{align}
\Pr(\| \bN\|_2\le x) =  F_1\left(\frac{ \frac{x^2}{\sigma^2}-\mu_{M,d}}{\sigma_{M,d}} \right) + \mathcal{O}(d^{-2/3}),\label{non-asym guanssian}
\end{align}
where the equality in \eqref{non-asym guanssian} follows from the fact that $f(x)=\mathcal{O}(g(x))$ if there exists a positive number $K$ and $x_0$ such that $|f(x)|\le Kg(x)$, $\forall x\ge x_0$ and $f(x), g(x): \R^+ \rightarrow \R^+$. The motivation of \cref{non-asym guanssian} is that we can apply it to \cref{theorem SOMP random L} for further analysis.
\end{Remark}


\begin{figure}
\centering
\includegraphics[width=0.6 \linewidth]{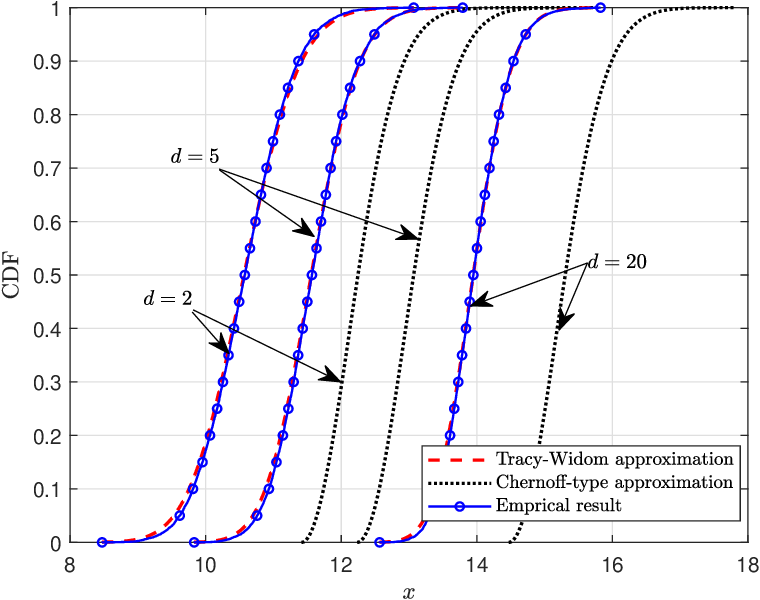}
\caption{Illustration for the approximation of $\text{Pr}(\|\bN\|_2 \le x)$ with Tracy-Widom distribution (asymptotic) in \eqref{guanssian dis} and Chernoff-type bound (non-asymptotic) in \eqref{bound exp N} ($\bN \in \R^{M \times d}, M=100, \sigma = 1$).} \label{check_TW}
\end{figure}

\begin{Remark} \label{expon N norm 2}
Different from \eqref{non-asym guanssian}, when the random matrix $\bN \in \R^{M \times d}$ has entries i.i.d. according to $\cN(0, \sigma^2)$, an alternative bound (Chernoff-type concentration bound) of non-asymptotic distribution of $\| \bN \|_2$ can be described by \cite{davidson2001local},
\begin{align}
\Pr\left(\frac{\| \bN \|_2}{\sigma} \le t+\sqrt{M} +\sqrt{d}\right) \ge 1- \exp(-t^2). \label{bound exp N}
\end{align}
\end{Remark}

While the bound in \eqref{bound exp N} is non-asymptotic, the bound in \eqref{bound exp N} is quite loose compared to the Tracy-Widom law distribution in \eqref{guanssian dis}. To observe it, we compare, in Fig.~\ref{check_TW}, the empirical CDF with the Tracy-Widom approximation in \eqref{guanssian dis} and the Chernoff-type bound in \eqref{bound exp N}. We can observe that Tracy-Widom approximation in \eqref{guanssian dis} is indeed accurate even when $d$ is small (e.g., $d=2$). It is obvious from Fig.~\ref{check_TW} that the bound in \eqref{bound exp N}  is quite loose. This reveals that the convergence in \eqref{guanssian dis} could kick in early when $d \geq 2$. In what follows, we treat \eqref{guanssian dis} as an accurate approximation and derive the SRP conditions.

 \begin{Remark} \label{corrolary SOMP gaussian}
From Remark \ref{second order accuracy} and Theorem \ref{theorem SOMP random L}, when mutual coherence $\mu < 1/(2L-1)$ and the entries of $\bN$ are i.i.d. with $\cN(0,\sigma^2)$, the lower bound of SRP of SOMPS in Algorithm \ref{alg_SOMP} satisfies
 \begin{align}
   P_s \ge F_1\left(\frac{{(1-(2L-1)\mu)^2 C_{\text{min}}^2}-4\sigma^2 \mu_{M,d}}{4\sigma^2 \sigma_{M,d}}\right) + \mathcal{O}(d^{-2/3}), \label{SOMPS guassian SRP}
  \end{align}
  where the inequality in \eqref{SOMPS guassian SRP} is a direct consequence of combining \eqref{SOMPS SRP} with \eqref{non-asym guanssian}.
 \end{Remark}

\subsubsection{SRP Guarantee of SOMPS Under Gaussian Noise}
In some applications, the SRP is desired to be close to one. The theorem below is an extension of Theorem \ref{theorem SOMP random L}, which identifies a condition for $C_{\text{min}}$ to guarantee the required SRP of SOMPS in Algorithm \ref{alg_SOMP}.
\begin{Corollary} \label{theorem required C}
For the signal model in \eqref{obOMP}, assume that mutual coherence $\mu$ of the measurement matrix $\bPhi$ satisfies $\mu < 1/(2L-1)$ and there exist positive $x_{\delta}$ and $\delta$ such that
\begin{align}
\Pr(\|  \bN \|_2 \le x_{\delta})  \ge  1-\delta, \label{delta eq}
\end{align}
where $0<\delta<1 $ is a small number and $x_{\delta}$ depends on $\delta$. If
\begin{align}
C_{\min} > \frac{2 x_{\delta}}{{1-(2L-1)\mu }},  \label{cor condition old}
\end{align}
SOMPS in Algorithm \ref{alg_SOMP}  can successfully recover the support set $\Omega$ with the probability exceeding  $1-\delta$.
\end{Corollary}
\begin{proof}
Similar to the proof of Theorem \ref{theorem SOMP random L}, we define the following two events,
\begin{align}
\mathcal{X}_1 &= \left\{\lA [\bC]_{i,:} \rA_2 > \frac{2 x_{\delta} }{{1-(2L-1)\mu }}, \forall i \in \Omega \right\} \nonumber \\
\mathcal{X}_2 &= \left\{ \| \bN\|_2 \le x_{\delta} \right\} \nonumber.
\end{align}
Then, according to Theorem \ref{theorem SOMP}, the SRP of SOMPS is lower bounded by
$\Pr(\mathcal{X}_1\cap \mathcal{X}_2 )\ge 1-\delta$.
\end{proof}

By Theorem \ref{theorem required C}, the row-sparse matrix $\bC$ in \eqref{obOMP} with $C_{\text{min}}$ satisfying the condition in \eqref{cor condition old} can guarantee the required SRP $1-\delta$.
In particular, when the entries in $\bN$ are i.i.d. according to $\cN(0,\sigma^2)$,
the performance guarantee of SOMPS in Algorithm \ref{alg_SOMP} can be identified as follows.

\begin{Corollary} \label{col required C gaussian}
Suppose $\delta>0 $ is a small number and {$F_1^{-1}(\cdot)$ is the inverse function of $F_1(\cdot)$}. If the noise $\bN$ in \eqref{obOMP} has entries i.i.d. according to $\cN(0,\sigma^2)$ and
\begin{align}
C_{\min} > \frac{2 \sqrt{(F_1^{-1}(1-\delta)\sigma_{M,d}+\mu_{M,d})\sigma^2}}{{1-(2L-1)\mu }}, \label{cor condition}
\end{align}
or equivalently
\begin{align}
\sigma < \frac{C_{\min}{(1-(2L-1)\mu )}  }{2 \sqrt{(F_1^{-1}(1-\delta)\sigma_{M,d}+\mu_{M,d})}}, \label{sigma condition SOMPS}
\end{align}
 then Algorithm \ref{alg_SOMP} successfully recovers the support set $\Omega$ with the probability exceeding $1-\delta + \mathcal{O}(d^{-2/3})$.
 \begin{proof}
From the second order accuracy of Tracy-Widom in Remark \ref{second order accuracy},
the following holds,
\begin{align}
\Pr\left(\|  \bN \|_2 \le \sqrt{(F_1^{-1}(1-\delta)\sigma_{M,d}+\mu_{M,d})\sigma^2}\right)\nonumber\\
 =1-\delta + \mathcal{O}(d^{-2/3}).\label{Pr inverse}
\end{align}
 Then, by letting $x_{\delta}= \sqrt{(F_1^{-1}(1-\delta)\sigma_{M,d}+\mu_{M,d})\sigma^2}$, we conclude that when the condition in \eqref{cor condition} holds,  Algorithm \ref{alg_SOMP} successfully recovers the support set $\Omega$ with the probability exceeding $1-\delta + \mathcal{O}(d^{-2/3})$.
 \end{proof}
\end{Corollary}

Observing \eqref{cor condition}, we find that it involves the inverse of the CDF of Tracy-Widom distribution and is difficult to gain intuitions on the relationship between $C_{\min}$ and the signal parameters such as $M,d$. The Chernoff-type bound in \eqref{bound exp N} can be alternatively used to gain useful intuitions when $\delta$ is close to zero.

\begin{Remark}\label{remark condition mu}
From \eqref{bound exp N}, when $x \ge (\sqrt{M} + \sqrt{d}+t)\sigma$, the following inequality
$\Pr(\| \bN \|_2 \le x) \ge 1-\exp(-t^2)$ holds.
Thus, from Theorem \ref{theorem required C}, the SRP is at least $ 1-\exp(-t^2)$ when
\begin{align}
C_{\text{min}} > \frac{2(\sqrt{M}+ \sqrt{d} +t )\sigma}{1-(2L-1)\mu}, \label{new require for C}
\end{align}
or equivalently
\begin{align}
\mu < \frac{1-2(\sqrt{M}+\sqrt{d} +t )\sigma/C_{\text{min}}}{2L-1}. \label{condition for mu}
\end{align}
Compared to the sufficient condition $\mu < 1/(2L-1)$ for the noiseless case \cite{MuOMP}, the upper bound of $\mu$ in \eqref{condition for mu} is smaller due to the noisy observations.
The difference between the bound in \eqref{condition for mu} and the condition in \cite{MuOMP} is proportional to the noise standard deviation.
\end{Remark}
\subsubsection{Effects of $M$ and $d$ on Performance Guarantee of SOMPS}

In the following, the effect of the number of measurements and the number of vectors on the condition in \eqref{new require for C} is of interest.
For the effect of $M$, first of all, we consider two extreme cases when $\sigma=0 $ and $\mu=0$  with fixed $C_{\text{min}}$.
The first case is ideal in terms of noise, and the second case is ideal in terms of MIP. We look into these extreme cases to identify if increasing the number of measurements helps or not:
\begin{itemize}
\item
For a measurement matrix $\bPhi$ with $\rank(\bPhi) = M$ and $\sigma=0$, as $M$ tends to $N$ while $N$ is fixed, the measurement matrix $\bPhi$ achieves the minimum mutual coherence satisfying $\mu < \frac{1}{2L-1}$.  In this case, increasing the number of measurements $M$ is advantageous.

\item When the noise level is fixed $\sigma > 0$ and  $\mu=0$, which can be achieved by $M \ge N$, the condition in \eqref{new require for C} becomes $C_{\text{min}} >{2(\sqrt{M}+ \sqrt{d} +t )\sigma}$. It means that $C_{\text{min}}$ should be proportional to the noise level $\sigma$. In particular, if we
    constrain such that $C_{\min}$ and $N$ are constant while increasing $M$, the recovery condition in \eqref{new require for C} with $\mu=0$ becomes infeasible eventually.
\end{itemize}

For general cases when $\mu\neq 0$ and $\sigma>0$, similar to the analysis in Remark \ref{condition of M}, the Welch bound in \eqref{bound mu 1}
can be incorporated into \eqref{new require for C} given that $\bPhi$ is designed to achieve the minimal mutual coherence \cite{str2003Grass,Herman2009High,elad_CS}.
For fixed $N$, we aim to analyze the performance guarantee by approaching $M$ to $N$.
We divide the discussions into two cases when (i) $C_{\text{min}}$ varies with $M$ and (ii) $C_{\text{min}}$ is fixed with $M$:
\begin{itemize}
\item
If  $C_{\text{min}}$ varies with $M$ such that $1- 2\sigma(\sqrt{M}+\sqrt{d}+t)/C_{\text{min}}$ is lower bounded
by some positive constant $\rho$,\footnote{In practice, one can adjust the power of the sparse signal, i.e., $C_{\text{min}}$, to achieve the requirement $1- 2\sigma(\sqrt{M}+\sqrt{d}+t)/C_{\text{min}} > \rho$.} the successful  recovery for SOMPS can be guaranteed with a high probability of $1-\exp(-t^2)$ if the number of measurements satisfies the following,
\begin{align} \label{require M SOMPS}
M > \frac{N}{\frac{{\rho}^2}{(2L-1)^2} (N-1)+1}.
\end{align}
The condition in \eqref{require M SOMPS} reveals the required number of measurements in terms of $N$, the sparsity level $L$, and the noise level involved in the constant $\rho$. It is consistent with the result in \eqref{bounded for M} for the bounded noise case.
\item
If $C_{\text{min}}$ is fixed, we incorporate the Welch bound in \eqref{bound mu 1} into the condition in \eqref{condition for mu}, leading to
\begin{align} \label{Cmin M}
C_{\text{min}} > \underbrace{\frac{2(\sqrt{M}+ \sqrt{d} +t )\sigma}{1-(2L-1)\sqrt{\frac{N-M}{M(N-1)}}}}_{\triangleq h(M)}.
\end{align}
Because both the numerator and denominator are increasing with $M,$ it is difficult to determine if the right-hand-side of \eqref{Cmin M} is increasing with respect to $M$. Taking the first order derivative of  $h(M)$ in \eqref{Cmin M} reveals that $\frac{dh(M)}{d M} < 0$. Thus, $h(M)$ is a decreasing function of $M$. In other words, increasing the number of measurements $M$ is advantageous in this case.
\end{itemize}


Intuitively, a larger number of sparse vectors $d$ can lead to a more accurate recovery performance. In order to evaluate the effect of $d$ on the recovery performance of SOMPS under the Gaussian noise, we have the following remark.

\begin{Remark} \label{condition d}
The value of $C_{\text{min}}$ is also a function of the number of sparse vectors $d$.
Without loss of generality, we assume here  $C_{\text{min}}^2$ is proportional to the number of sparse vectors $d$, i.e., $C_{\text{min}}^2 = dc_{\text{m}}$, where $c_\text{m}$ is a constant. Then, the condition in \eqref{condition for mu} can be rewritten as with respect $d$. Specifically,  the SRP of Algorithm \ref{alg_SOMP} is at least $ 1-\exp(-t^2)$ when
\begin{align}
d> \frac{(\sqrt{M}+t)^2}{\left(\sqrt{c_m}\frac{1-(2L-1)\mu}{2\sigma}-1\right)^2}. \label{approx d bound SOMPS}
\end{align}
If the noise standard deviation $\sigma$ is small or alternatively, $c_m$ is large such that $\sqrt{c_m}\frac{1-(2L-1)\mu}{2\sigma}\gg 1$,
the condition \eqref{approx d bound SOMPS} can be approximated to
\begin{align}
d &>\frac{4(\sqrt{M}+t)^2\sigma^2}{{(1-(2L-1)\mu)^2}{c_m}}. \nonumber
\end{align}
This reveals the fact that the required number of sparse vectors $d$ is proportional to the noise variance $\sigma^2$ in order to guarantee the successful support recovery. Because the mutual coherence $\mu$ depends on the number of measurements $M$, i.e.,  the larger the $M$ value is, the smaller the $\mu$ value will be, the relationship between  $M$ and $d$ is not expressed explicitly in \eqref{approx d bound SOMPS}. Nevertheless, through the simulation in Section \ref{Ssimulation}, one can find more measurements lead to a fewer number of sparse vectors to guarantee the successful recovery.
\end{Remark}

In summary, for a given system, the SRP of SOMPS in Algorithm \ref{alg_SOMP} is provided in \cref{SOMPS guassian SRP}, which serves as a quantitative measure to evaluate the recovery performance of the system.
Moreover, our theoretical findings also provide the
guidelines for adjusting the system parameters to achieve the desired recovery performance.
For example, if it is required that the SRP of SOMPS should be larger than 0.99,
the signal power $C_{\min}$ can be adjusted based on \cref{cor condition} or \cref{new require for C}, the dimension of the measurement matrix $\bPhi$ can be modified  using \cref{Cmin M}, and the number of sparse vectors can be determined through \cref{approx d bound SOMPS}.
By following these guidelines, it becomes possible to finely tune the system parameters and achieve specific performance objectives.
\subsection{Analysis of SOMPT Under Unbounded Noise} \label{sec:unboundSOMPT}
In this subsection, we discuss the performance of SOMPT in Algorithm \ref{alg_SOMP_stop} when $\bN$ is unbounded noise.
The procedure is very similar to that of SOMPS, and we present only the main results while simplifying the proofs and discussions.
First of all, the following theorem shows the lower bound of the SRP.
\begin{Theorem} \label{theorem SOMP random}
 Given the signal model in \eqref{obOMP} and $\mu < 1/(2L-1)$ with $\mu$ being mutual coherence of  $\bPhi$,  if $\bN$ satisfies
$\Pr(\|  \bN \|_2 \le x) =  F_{N}(x) $
and the threshold value is set to
\begin{align}
\tau = \frac{C_{\text{min}}(1-(2L-1)\mu)}{2},  \nonumber
\end{align}
the SRP of Algorithm \ref{alg_SOMP_stop} is lower bounded by
 \begin{align}
   P_s \ge F_{N}\left(\frac{C_{\text{min}}{{(1-(2L-1)\mu )}} } {2 }\right), \label{prob tau random}
  \end{align}
  where $C_{\min} = \min\limits _{i\in \Omega}  \lA [\bC]_{i,:} \rA_2$.
\end{Theorem}
\begin{proof}
By Theorem \ref{theorem SOMP stop},  SOMPT in Algorithm \ref{alg_SOMP_stop} guarantees the support recovery when
$
\lA [\bC]_{i,:} \rA_2 > \frac{2 x }{{1-(2L-1)\mu }}, \forall i \in \Omega,
\| \bN\|_2 \le x$, and $\tau = x$.
Then, following the same procedure as the proof of Theorem \ref{theorem SOMP random L}, the SRP of Algorithm \ref{alg_SOMP_stop} is lower bounded by $F_N(x), \forall x>0$.
Substituting $x={C_{\text{min}}(1-(2L-1)\mu)}/{2}$ leads to \eqref{prob tau random}.
\end{proof}

 \begin{Remark} \label{col SOMP gaussian stop}
When the entries in $\bN$ are i.i.d. Gaussian with $\cN(0,\sigma^2)$ and the threshold value of SOMPT is given by
 \begin{align*}
\tau = \frac{C_{\text{min}}{(1-(2L-1)\mu )}}{2}, \nonumber
\end{align*}
then the SRP of SOMPT is lower bounded by
\begin{align}
   P_s \ge F_1\left(\frac{{(1-(2L-1)\mu)^2 C_{\text{min}}^2}-4\sigma^2 \mu_{M,d}}{4\sigma^2 \sigma_{M,d}}\right) + \mathcal{O}(d^{-2/3}), \label{SOMPT guassian SRP}
  \end{align}
  where the result in \eqref{SOMPT guassian SRP} is a direct consequence of combining \eqref{non-asym guanssian}  with \eqref{prob tau random}.
  \end{Remark}

The following corollary identifies the condition of $C_{\min}$ to guarantee the required SRP of SOMPT in Algorithm \ref{alg_SOMP_stop}.
\begin{Corollary} \label{the required C stop}
Given the signal model in \eqref{obOMP} and $\mu < 1/(2L-1)$ with $\mu$ being mutual coherence of  $\bPhi$, assume that $x_{\delta}$ and $\delta$ satisfy
$\Pr(\|  \bN \|_2 \le x_{\delta})  \ge 1-\delta$. Then, if the threshold value of SOMPT is $\tau = x_{\delta}$ and
\begin{align}
C_{\min} > \frac{2 x_{\delta}}{{1-(2L-1)\mu }} \nonumber
\end{align}
with  $C_{\min} = \min\limits _{i\in \Omega}  \lA [\bC]_{i,:} \rA_2$, the SOMPT in Algorithm \ref{alg_SOMP_stop} can successfully recover the support $\Omega$ with the probability exceeding $1-\delta$.
\end{Corollary}
\begin{proof}
We define the following two events,
\begin{align}
\mathcal{X}_1 &= \left\{\lA [\bC]_{i,:} \rA_2 > \frac{2 x_{\delta} }{{1-(2L-1)\mu }}, \forall i \in \Omega \right\} \nonumber \\
\mathcal{X}_2 &= \left\{ \| \bN\|_2 \le x_{\delta} \right\} \nonumber.
\end{align}
Then, according to Theorem \ref{theorem SOMP stop}, if the threshold value $\tau=x_{\delta}$, the SRP of SOMPT is lower bounded by $\Pr(\mathcal{X}_1\cap \mathcal{X}_2 )\ge 1-\delta$. This concludes the proof.
\end{proof}

When the noise matrix $\bN$ has i.i.d. Gaussian entries,   the following corollary characterizes the performance guarantee of SOMPT under the Gaussian noise.
\begin{Corollary} \label{col required C gaussian stop}
We let the threshold value of SOMPT  be
 \begin{align}
 \tau = \sqrt{(F_1^{-1}(1-\delta)\sigma_{M,d}+\mu_{M,d})\sigma^2}. \label{stop SOMPT}
 \end{align}
Then, if the following holds,
\begin{align}
C_{\min} > \frac{2 \sqrt{(F_1^{-1}(1-\delta)\sigma_{M,d}+\mu_{M,d})\sigma^2}}{{1-(2L-1)\mu }}, \label{cor condition stop}
\end{align}
or similarly
\begin{align}
\sigma < \frac{C_{\min}{(1-(2L-1)\mu )}  }{2 \sqrt{(F_1^{-1}(1-\delta)\sigma_{M,d}+\mu_{M,d})}}, \label{sigma condition SOMPT}
\end{align}
then SOMPT in Algorithm \ref{alg_SOMP_stop} can successfully recover the support set $\Omega$ with probability higher than $1-\delta +  \mathcal{O}(d^{-2/3})$.
\begin{proof}
According to Remark \ref{second order accuracy},
the following holds,
\begin{align*}
\Pr\left(\|  \bN \|_2 \le \sqrt{(F_1^{-1}(1-\delta)\sigma_{M,d}+\mu_{M,d})\sigma^2}\right)\nonumber\\
 =1-\delta + \mathcal{O}(d^{-2/3}).
\end{align*}
 Then, by letting $x_{\delta}= \sqrt{(F_1^{-1}(1-\delta)\sigma_{M,d}+\mu_{M,d})\sigma^2}$ in Theorem \ref{the required C stop}, it is concluded that when the guarantee condition in \eqref{cor condition stop} or \eqref{sigma condition SOMPT} holds,  SOMPT can successfully recover the support set $\Omega$ with probability exceeding $1-\delta +  \mathcal{O}(d^{-2/3})$.
\end{proof}
\end{Corollary}

\begin{Remark} \label{condition d sompt}
According to Corollary \ref{col required C gaussian stop}, if the threshold value of the SOMPT is given by \eqref{stop SOMPT}, we can obtain the conditions of $C_{\min}$ and
$\sigma$ on which SOMPT can guarantee the required SRP of the support set.
Moreover, the conditions of $M$ and $d$ for the performance guarantee can also be calculated in a similar way as in Remark \ref{remark condition mu} and Remark \ref{condition d},
\begin{align}
M > \frac{N}{\frac{{\rho}^2}{(2L-1)^2} (N-1)+1}. \label{condition for M SOMPT}
\end{align}
and
\begin{align}
d &>\frac{(\sqrt{M}+t)^2}{\left(\sqrt{c_m}\frac{1-(2L-1)\mu}{2\sigma}-1\right)^2} . \label{approx d bound SOMPT}
\end{align}
\end{Remark}

Similar to SOMPS in \cref{sec:SOMPS}, for a given system, SOMPT in Algorithm \ref{alg_SOMP_stop} also has an SRP  as provided in \cref{SOMPT guassian SRP}. Guidelines for adjusting parameters such as signal power, number of measurements, and number of sparse vectors are given in \cref{cor condition stop}, \cref{condition for M SOMPT}, and \cref{approx d bound SOMPT} respectively.

\section{Simulation Results} \label{Ssimulation}

{
In this section, we conduct numerical simulations to validate the main results of SOMPS and SOMPT in Section \ref{section bound noise} and Section \ref{section unbounded noise}. In the bounded noise case, the performance of the SOMPS and SOMPT is determined by the parameters $\mu$, $C_{\min}$, $L$, and $\epsilon$.
In the unbounded noise case, this is determined by $\mu$, $C_{\min}$, $L$, $d$, and $\sigma$. For each set of simulation parameters, we perform $10,000$-trials Monte-Carlo simulation. In the trials, the measurement matrix $\bPhi$ is deterministic, which is generated by following the methodology described in \cite{elad_CS}. The approach in \cite{elad_CS} minimizes the average measure of the mutual coherence of the measurement matrix, which has been shown to provide better sparse reconstruction performance. In each trial, the support of the sparse signal $\bC$ and the noise $\bN$ are generated randomly and independently. After randomly generating the support set,
we assign the
non-zero components in $\bC$ to be $z_{i,j}\sqrt{C_{\min}^2/d}, \forall i\in\Omega, j\le d$, where  the random variable $z_{i,j}$ is either $-1$ or $1$ with
probability $0.5$.}

\subsection{SRP Validation } \label{section simulation SRP}
\subsubsection{Bounded Noise}

\begin{figure}[!t]
\centering
\includegraphics[width=0.6 \linewidth]{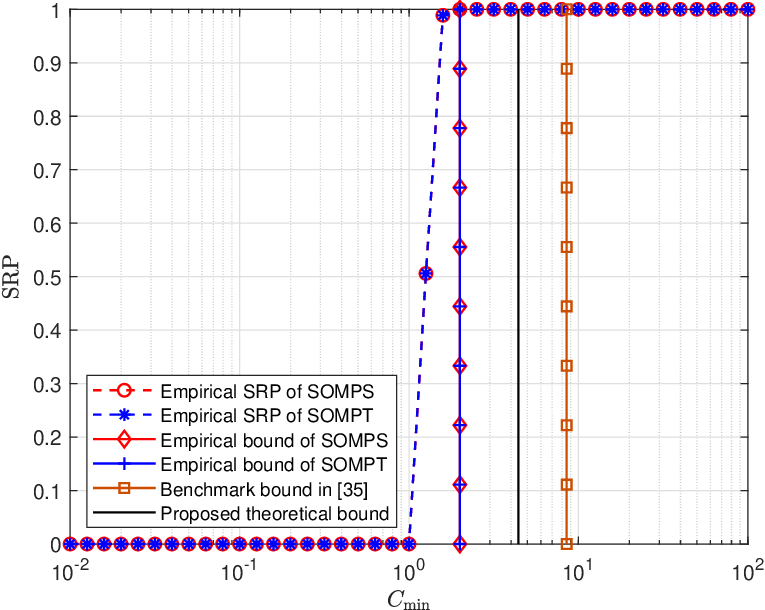}
\caption{SRP of SOMP versus $C_{\min}$ under bounded noise ($M=100, N=200,L=4,d=4, \|\bN \|_2 \le 1$).}  \label{figure_Cmin_bounded_guarantee}
\end{figure}

In Fig. \ref{figure_Cmin_bounded_guarantee}, we evaluate the SRP performances of SOMPS and SOMPT under the bounded noise based on empirical simulations as well as the theory established in Section IV. In Fig. \ref{figure_Cmin_bounded_guarantee}, we also illustrate the theoretical bound of $C_{\min}$ provided in \cite{Wang2013Performance} for a comparison.
For the numerical simulations, the simulation parameters are set to $M=100, N=200,L=4$, and $d=4$.
From Theorem \ref{theorem SOMP} and Theorem \ref{theorem SOMP stop}, when the $C_{\min} > \frac{2 \epsilon}{{1-(2L-1)\mu }}$, SOMPS and SOMPT can successfully recover the support set.
Thus, the "Proposed theoretical bound" in Fig. \ref{figure_Cmin_bounded_guarantee} is plotted at $C_\text{min}$ values through  the
formula $C_{\min} =\frac{2 \epsilon}{{1-(2L-1)\mu }}$ according to \cref{theorem bounded noise alg1}.
The curves of the $``\text{Empirical SRP}"$ are obtained by averaging the trials and depict the empirical SRP for different values of $C_{\min}$. The curves of the $``\text{Empirical bound}"$ are drawn at $C_\text{min}$ values where the successful recovery for all trials is achieved.
For a comparison, we also plot the theoretical bound for $C_{\min}$ in \cite{Wang2013Performance}  as a benchmark, denoted as "Benchmark bound in \cite{Wang2013Performance}".
As seen from Fig. \ref{figure_Cmin_bounded_guarantee},  when the sufficient conditions in Theorem \ref{theorem SOMP} and Theorem \ref{theorem SOMP stop} are satisfied, SOMPS and SOMPT both guarantee the successful support recovery.
This verifies that the simulation results are consistent with the derived theoretical bounds.
Upon examining the results in \cref{figure_Cmin_bounded_guarantee}, it is evident that the proposed theoretical bound for $C_{\min}$ to ensure successful recovery is tighter than the one presented in \cite{Wang2013Performance}. This improvement in tightness is attributed to the $\ell_2$-norm of noise involved in our analysis.
Moreover, it is interesting to find that SOMPS and SOMPT can achieve similar SRP when the noise is bounded. 
This is because, in the case of bounded noise, SOMPT's stopping criterion $\| \bR^{(l)}\|_2 < \epsilon$ aligns well with the condition that the noise is bounded, i.e., $\|\bN\|_2 \le \epsilon$. As a result, the stopping criterion of $\| \bR^{(l)}\|_2 < \tau$ for SOMPT allows it to correctly terminate the iteration process.

The gap between the theoretical bounds and the empirical bounds in Fig. \ref{figure_Cmin_bounded_guarantee} can be interpreted as follows. We let $\cC_0$ and ${\cC}_1$ be, respectively, the set of all sparse signals $\bC$ and the set of $10,000$ realizations of $\bC$ (thus, ${\cC}_1 \subset \cC_0$) that are used to evaluate the Empirical SRP and Empirical bounds. It is worth noting that the theoretical bounds in Fig. \ref{figure_Cmin_bounded_guarantee} are valid for all $\bC \in \cC_0$, while the empirical bounds are plotted only for $\bC \in \cC_1$. Depending on the realization of sparse signals and noise, some sparse signals are difficult to be recovered, which require higher $C_{\text{min}}$ or lower noise level $\epsilon$ to guarantee the successful recovery. Therefore, it is possible that there exists any $\bC\in \cC_0 \cap \cC_1^c$ such that the empirical bounds of such $\bC$ are close to the theoretical bounds in Fig. \ref{figure_Cmin_bounded_guarantee}. This accounts for the reason that there is a gap between the theoretical bound and the empirical bounds. The same reasons also account for theoretical gap observed in the simulations of  the following subsections.

\begin{figure}[tb]
\centering
\includegraphics[width=0.6 \linewidth]{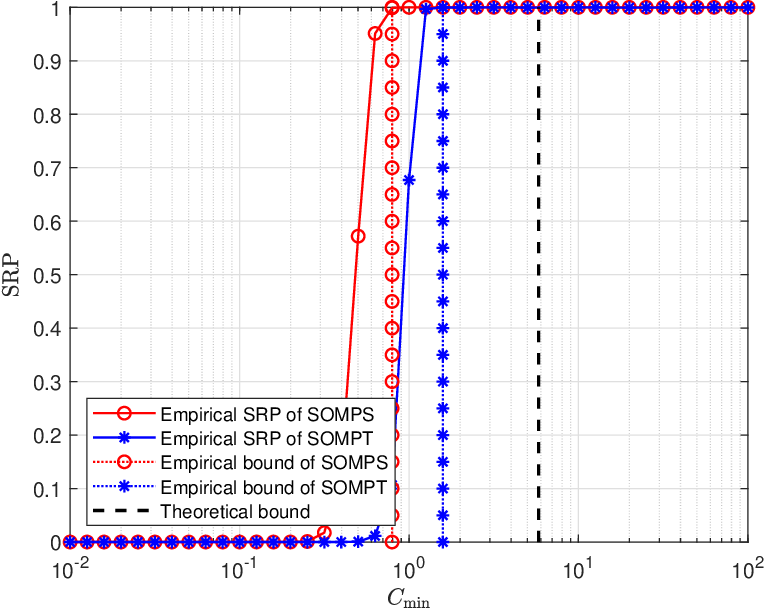}
\caption{SRP of SOMP versus $C_{\min}$ under Gaussian noise ($M=100, N=200,L=4,d=4, [\bN]_{i,j} \sim \cN(0,1/M) ,\delta=10^{-3}$).} \label{figure_Cmin_Guassian_guarantee}
\end{figure}

\subsubsection{Gaussian Noise}

In Fig. \ref{figure_Cmin_Guassian_guarantee}, we validate our analysis when the noise is random Gaussian.
The simulation parameters are set to $M=100, N=200, L=4,d=4$, and $\delta=10^{-3}$.
In Fig. \ref{figure_Cmin_Guassian_guarantee}, we illustrate the simulated SRPs of SOMPS and SOMPT with different $C_{\min}$, and compare them with the proposed theoretical bounds.
Similar to \cref{figure_Cmin_bounded_guarantee}, the curves of the ``Empirical SRP'' in \cref{figure_Cmin_Guassian_guarantee} are obtained by averaging the trials, the curve of "Theoretical bound" is plotted at the value of $C_{\min}$ which satisfies the condition in \eqref{cor condition} and \eqref{cor condition stop}, and the curves of the ``Empirical bound'' are drawn at $C_{\min}$ values at which
the recovery probability of the trials is at least $1-\delta$. Therefore, the curves for the empirical and theoretical bounds are perpendicular to the axis.
 As we can see from \cref{figure_Cmin_Guassian_guarantee}, when the $C_{\min}$ satisfies the condition in \eqref{cor condition} and \eqref{cor condition stop},
the recovery probability is at least $1-\delta$, which is consistent with the empirical bounds in Fig. \ref{figure_Cmin_Guassian_guarantee}.
 Unlike the scenario of bounded noise, where SOMPS and SOMPT exhibit the similar performance, the SOMPT in \cref{figure_Cmin_Guassian_guarantee} does not yield same performance as SOMPS in the scenario of unbounded noise. 
 This is because  a priori knowledge of the sparsity level helps to improve the recovery performance of SOMPS, and the selection of an appropriate threshold value for SOMPT is more complex in the scenario of unbounded noise. Considering the empirical performance of the SOMPT depends on the selection of the threshold value, the optimization of this threshold is an interesting future work.

Compared to the bounded noise scenario in Fig. \ref{figure_Cmin_bounded_guarantee}, the theoretical bounds in Fig. \ref{figure_Cmin_Guassian_guarantee} seem looser. This is because, apart from the difference between $\cC_0$ and $\cC_1$,  the empirical bounds in Fig. \ref{figure_Cmin_Guassian_guarantee} are plotted by averaging the simulated sparse signals in $\cC_1$. Therefore, in addition to the case when $\bC\in \cC_0 \cap \cC_1^c$, it is also possible that there exists some sparse signal $\bC\in \cC_1$ that shows poorer empirical results, i.e., requiring higher $C_{\text{min}}$ or lower $\sigma$ than the average result in Fig. \ref{figure_Cmin_bounded_guarantee}. It is for this reason that the derived bound is not tight enough. The same observation was also made previously in a separate study \cite{Ben2010Coherence} that analyzed the recovery performance of OMP.

\begin{figure}[!t]
\centering
\includegraphics[width=0.6 \linewidth]{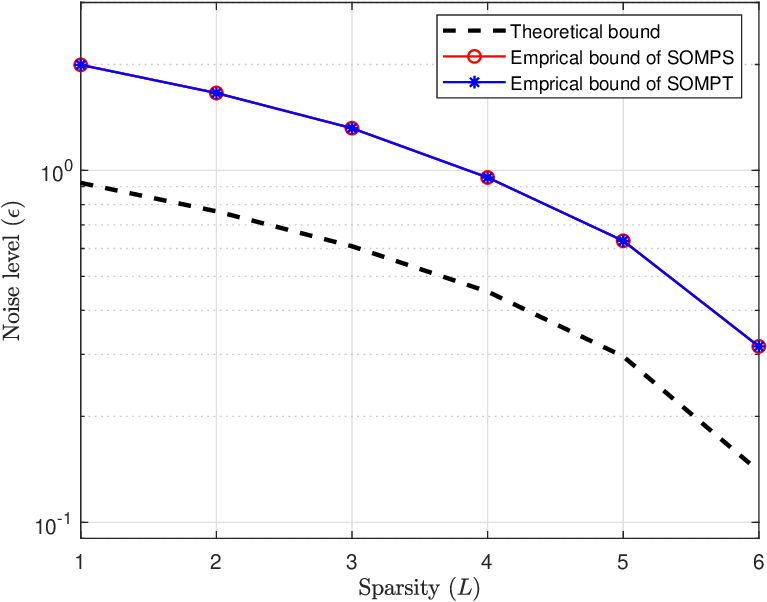}
\caption{Noise level for guarantee of SOMP  versus sparsity $L$ under bounded noise ($M=100, N=200,  d=4,C_{\min}=2,\|\bN \|_2 \le \epsilon $).} \label{figure_L_Bounded}
\end{figure}

\begin{figure}[!t]
\centering
\includegraphics[width=0.6 \linewidth]{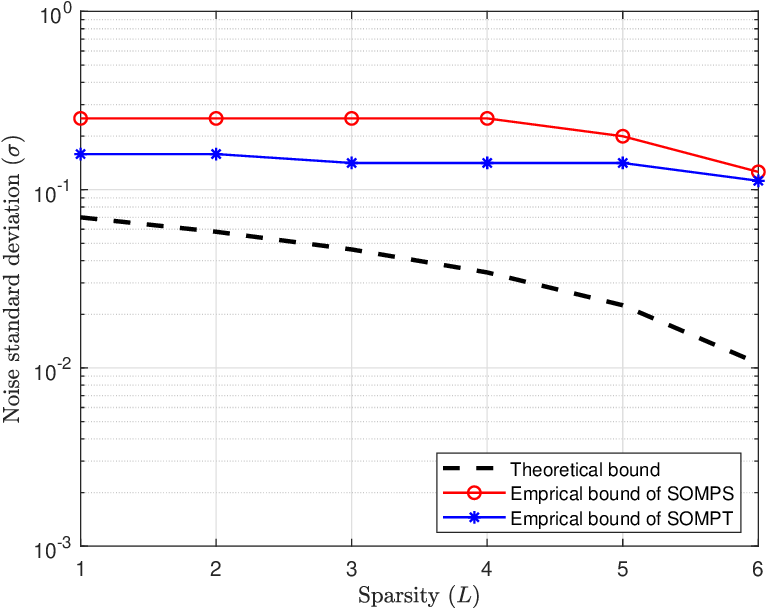 }
\caption{Noise level for guarantee of SOMP versus sparsity $L$ under Gaussian noise ($M=100, N=200, d=4,C_{\min}=2, [\bN]_{i,j} \sim \cN(0,\sigma^2) ,\delta=10^{-3}$).} \label{figure_L_1_Gaussian}
\end{figure}

\subsection{Effects of Various Factors on SRP Guarantees}
\subsubsection{The Sparsity  Level $L$} \label{section VIB1}

{
In Fig. \ref{figure_L_Bounded}, we illustrate our analysis for bounded noise under different levels of sparsity.
Based on the methodology in \cite{elad_CS},
the generated measurement matrix has the mutual coherence $\mu=0.0782$, and then, the maximal sparsity level to satisfy $\mu<1/(2L-1)$ is given by $L=6$.
When $ L > 6$, the proposed analysis cannot provide a valid bound.
 Thus, we let $L\le 6$ in Fig. \ref{figure_L_Bounded}.
The theoretical bound in Fig. \ref{figure_L_Bounded} is the maximum noise level across different signal sparsity levels computed according to Theorem \ref{theorem SOMP} and Theorem \ref{theorem SOMP stop}.
The empirical curves are the simulated noise levels that guarantee the exact support recovery at each sparsity level.
As we can see, with the increasing sparsity level, the empirical curves should decrease to guarantee the support recovery for SOMPS and SOMPT,
which corresponds to the theoretical bounds in Theorem \ref{theorem SOMP} and Theorem \ref{theorem SOMP stop}.}

In Fig. \ref{figure_L_1_Gaussian}, we illustrate our analysis for Gaussian noise under different levels of sparsity.
The simulation parameters are $M=100, N=200,d=4$, $\delta=10^{-3}$, and $L=1,2,\ldots,6$.
{According to \eqref{sigma condition SOMPS} and \eqref{sigma condition SOMPT}, the noise level $\sigma$ that guarantees the recovery performance of SOMPS and SOMPT decreases as sparsity $L$ increases, which is indicated by the dashed curve of the theoretical bounds.} {Note in Fig. \ref{figure_L_1_Gaussian} that when the sparsity level is large, the gap between the theoretical bound and the empirical curve increases. This
 is because if the $L$ is large, the denominators $1-(2L-1)\mu $ in \eqref{cor condition} and \eqref{cor condition stop} will be close to zero, then the bound becomes loose.}

\subsubsection{The Number of Sparse Vectors $d$}

In Figs. \ref{figure_d_1}-\ref{figure_d_2}, we show the SRPs of SOMPS and SOMPT under Gaussian noise for different values of $d$, where $C_{\text{min}}^2$ is proportional to the number of sparse vectors $d$, i.e., $C_{\text{min}}^2 =dc_m$ with fixed $c_m=1$.
For the simulation, we use the colormap to define the SRP in the range of $[0,1]$, where the case $\text{SRP}=0$ is in black and the case $\text{SRP}=1$ is in white. The theoretical bounds for the noise level $\sigma$ are based on expressions in \eqref{sigma condition SOMPS} for SOMPS and \eqref{sigma condition SOMPT} for SOMPT.
As can be seen from Figs. \ref{figure_d_1}-\ref{figure_d_2}, as the noise level increases, the required number of sparse vectors for both SOMPS and SOMPT increases accordingly, which is consistent with the derived theoretical bound.

\begin{figure}[!t]
\centering
\includegraphics[width=0.6 \linewidth]{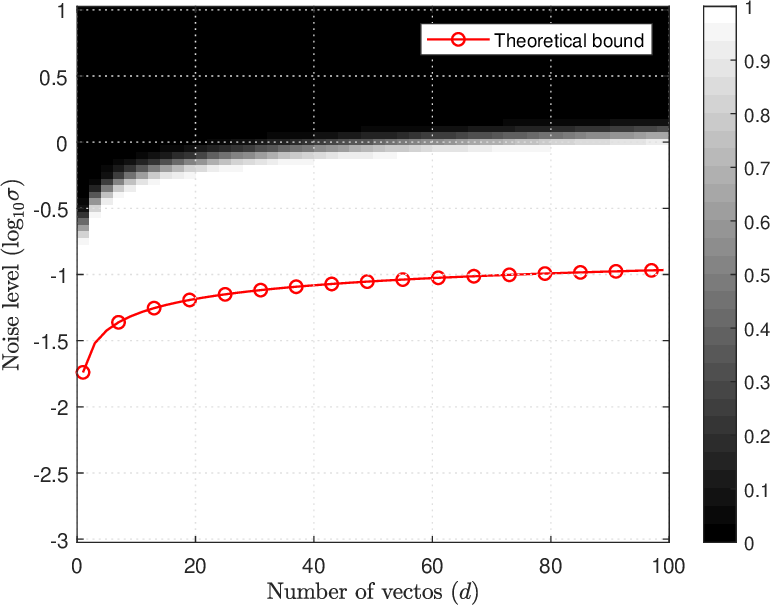}
\caption{SRP of SOMPS (Algorithm 1) with different $d$ ($M=100, N=200,L=4,C_{\text{min}} =\sqrt{d}, [\bN]_{i,j} \sim \cN(0,\sigma^2),\delta=10^{-3} $).} \label{figure_d_1}
\end{figure}

\begin{figure}[!t]
\centering
\includegraphics[width=0.6 \linewidth]{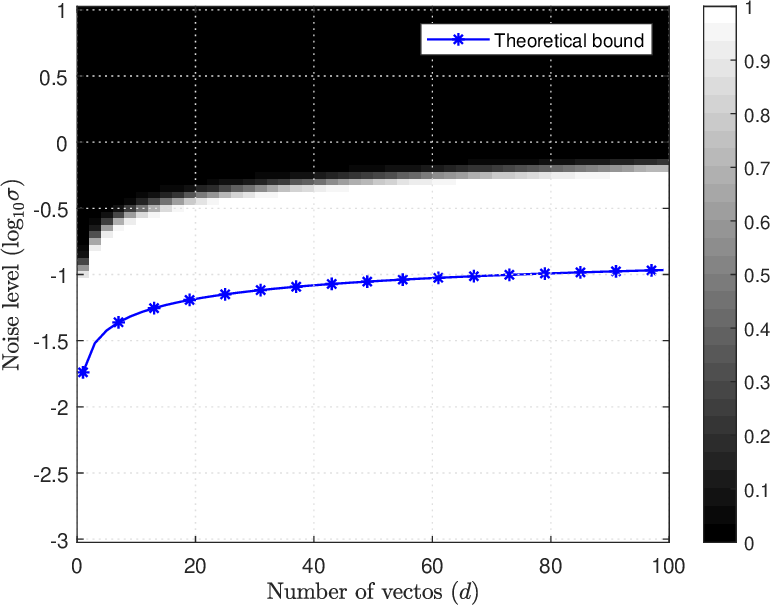}
\caption{SRP of SOMPT (Algorithm 2) with different $d$ ($M=100, N=200,L=4,  C_{\text{min}} =\sqrt{d}, [\bN]_{i,j} \sim \cN(0,\sigma^2),\delta=10^{-3} $).} \label{figure_d_2}
\end{figure}

\subsubsection{The Number of Measurements $M$}\label{section simulation M}


In  Figs. \ref{figure_Gaussian_M_3}-\ref{figure_Gaussian_M_4}, we illustrate our analysis by evaluating the number of measurements $(M)$ across the sparsity level $(L)$ under Gaussian noise. The simulation parameters are $N=200, d=4, \sigma=0.01$, and $\delta=10^{-3}$. The theoretical bounds in Figs. \ref{figure_Gaussian_M_3}-\ref{figure_Gaussian_M_4} for each sparsity level are obtained by increasing the number of measurements $M$ until the conditions \eqref{cor condition} and \eqref{cor condition stop} are satisfied. Observing the simulations in Figs. \ref{figure_Gaussian_M_3}-\ref{figure_Gaussian_M_4}, as the sparsity level $L$ increases, one can find the required number of measurements to guarantee the successful recovery also increases, which aligns with our derived theoretical bounds. Interestingly, one can find in Figs. \ref{figure_Gaussian_M_3}-\ref{figure_Gaussian_M_4} that the theoretical bound is tight when the sparsity level is small, which verifies our discussions in Remark \ref{condition of M}. Moreover, as the values of $M$ and $L$ increase, the difference between $\cC_0$ and ${\cC}_1$ increases, therefore we can observe that the theoretical bounds become loose in Figs.~\ref{figure_Gaussian_M_3}-\ref{figure_Gaussian_M_4}.

\begin{figure} [!t]
\centering
\includegraphics[width=0.6 \linewidth]{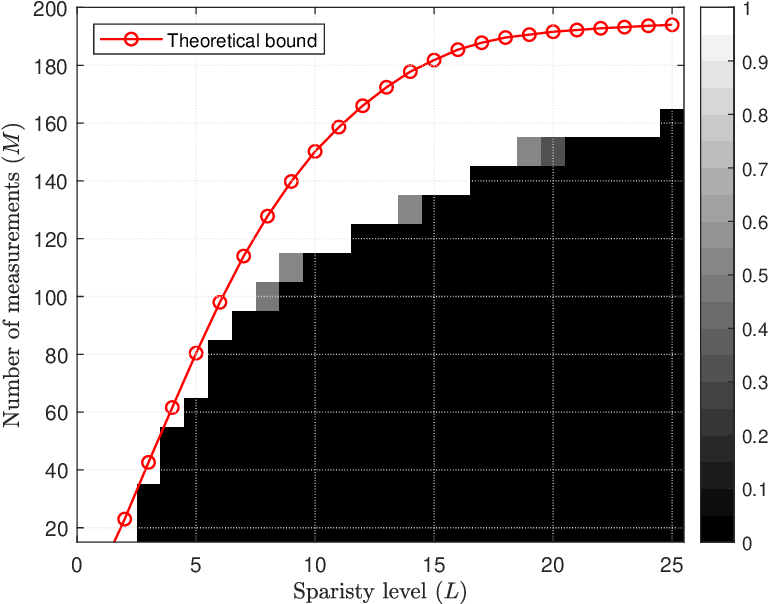}
\caption{SRP of SOMPS under Gaussian noise with different $M$ ($N=200,C_{\text{min}} =2, d=4,  [\bN]_{i,j} \sim \cN(0,\sigma^2),\sigma=0.1 ,\delta=10^{-3}$).} \label{figure_Gaussian_M_3}
\end{figure}

\begin{figure}[!t]
\centering
\includegraphics[width=0.6 \linewidth]{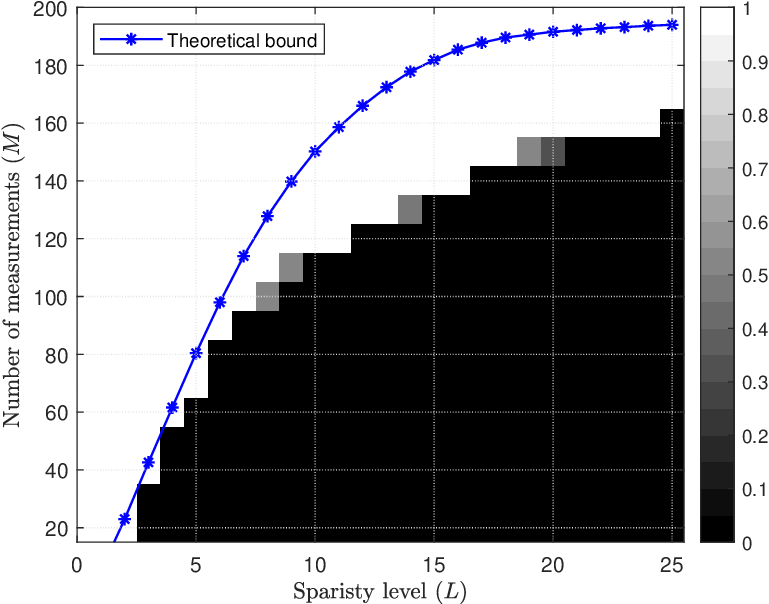}
\caption{SRP of SOMPT under Gaussian noise with different $M$ ($N=200,C_{\text{min}}=2, d=4,  [\bN]_{i,j} \sim \cN(0,\sigma^2),\sigma=0.1 ,\delta=10^{-3}$).} \label{figure_Gaussian_M_4}
\end{figure}

In  Fig. \ref{figure_Gaussian_M_D}, we illustrate our analysis of SOMP  under Gaussian noise by evaluating the number of measurements $M$ with the number of sparse vectors $d$, where $C_{\text{min}}^2$ is proportional to $d$, i.e., $C_{\text{min}}^2 =dc_m$ with fixed $c_m=2$. The remaining simulation parameters are $N=200, L=4, [\bN]_{i,j} \sim \cN(0,\sigma^2), \sigma=0.02$, and $\delta=10^{-3}$. In Fig. \ref{figure_Gaussian_M_D}, the theoretical bound of the number of measurements $M$ for each $d$ is the minimal $M$ which satisfies the conditions in \eqref{cor condition} for SOMPS and \eqref{cor condition stop} for SOMPT. One can find from the theoretical and empirical bounds in \cref{figure_Gaussian_M_D} that with the increasing of the number of vectors, the required number of measurements to guarantee the successful recovery decreases. Meanwhile, the number of required measurements stabilizes when the number of sparse vectors is sufficiently large. 
Observing \cref{figure_Gaussian_M_D}, we can also note that when $d$ is small, there exists a large gap between the theoretical curve and the empirical bound. This discrepancy arises because a smaller value of $d$ corresponds to a lower signal power $C_{\min}=\sqrt{2d}$. Consequently, a significantly larger number of measurements is required to ensure successful recovery, leading to a relatively loose theoretical bound.
Nevertheless, the results in \cref{figure_Gaussian_M_D} align with the approximate analysis presented in \cref{condition d}, which explores the relationship between $M$ and $d$.


\begin{figure} [!t]
\centering
\includegraphics[width=0.6\linewidth]{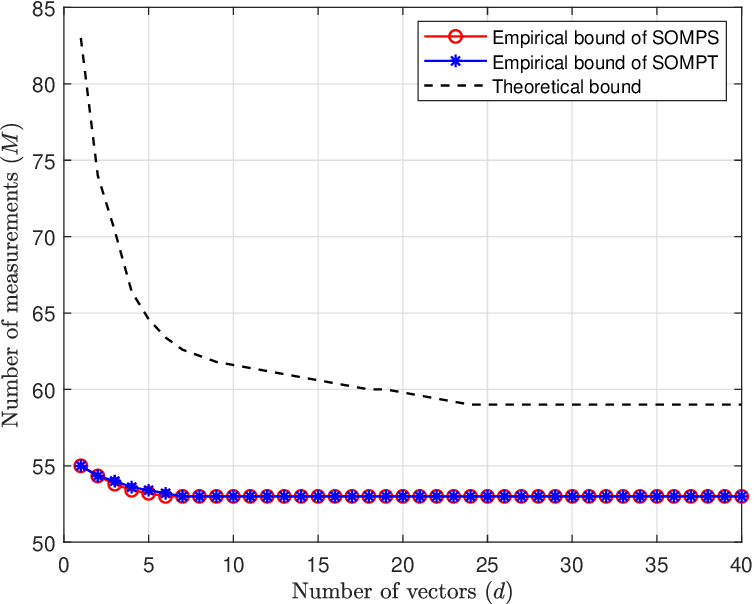}
\caption{Performance guarantee of SOMP for different $M$ and $d$ under Gaussian noise ($N=200,C_{\text{min}} =\sqrt{2d}, L=4, [\bN]_{i,j} \sim \cN(0,\sigma^2),\sigma=0.02 ,\delta=10^{-3} $).} \label{figure_Gaussian_M_D}
\end{figure}

\subsection{The Effect of Dynamic Range on SRP } \label{section DR}
For the previous simulation setting,
the values of $\| [\bC]_{i,:}\|_2$ for different $i$ are the same.  In Fig. \ref{figure_DR_Gaussian}, we evaluate the effect of dynamic
range of the entries of sparse matrix $\bC$  on the recovery performance of SOMP. The simulation parameters are $M=100, N=200, L=4,d=4$, and $\delta=10^{-3}$. We recall that the dynamic range for the case of OMP, which is defined by the minimal and maximal absolute value of the non-zero elements of the sparse vector, affects the recovery performance of OMP \cite{Mi2017Prob,Lee2020Error}. Accordingly, the dynamic range of the sparse signal for MMV is defined as $C_{\min} = \min_{i\in \Omega} \|[\bC]_{i,:}\|_2$ and $C_{\max} = \max_{i\in \Omega} \|[\bC]_{i,:}\|_2$. For the simulation setting, we first generate $L$ random variables $C_i , \forall i \in \Omega$ from a uniform distribution on the interval $[C_{\min},C_{\max}]$, and then the non-zero elements of $[\bC]_{i,j}, \forall i \in \Omega, j\le d$ are drawn from a white Gaussian distribution with proper normalization to satisfy $\|[\bC]_{i,:}\|_2=C_i, \forall i \in \Omega$. In Fig. \ref{figure_DR_Gaussian}, the empirical curves illustrate the noise levels which guarantee the recovery of SOMP at each dynamic range, i.e., ${C_{\min}}/{C_{\max}}$. The theoretical bound in Fig. \ref{figure_DR_Gaussian} is plotted according to \eqref{sigma condition SOMPS} and \eqref{sigma condition SOMPT}. From Fig. \ref{figure_DR_Gaussian}, as the ${C_{\min}}/{C_{\max}}$ decreases, the empirical recovery performance of SOMP deteriorates under both bounded and Gaussian noises, which is consistent with the sensitivity of OMP \cite{Mi2017Prob,Lee2020Error}. Meanwhile, one can find that the derived theoretical bounds still reveal the trend that follows the trends of the empirical curves.

\begin{figure}[!t]
\centering
\includegraphics[width=0.6\linewidth]{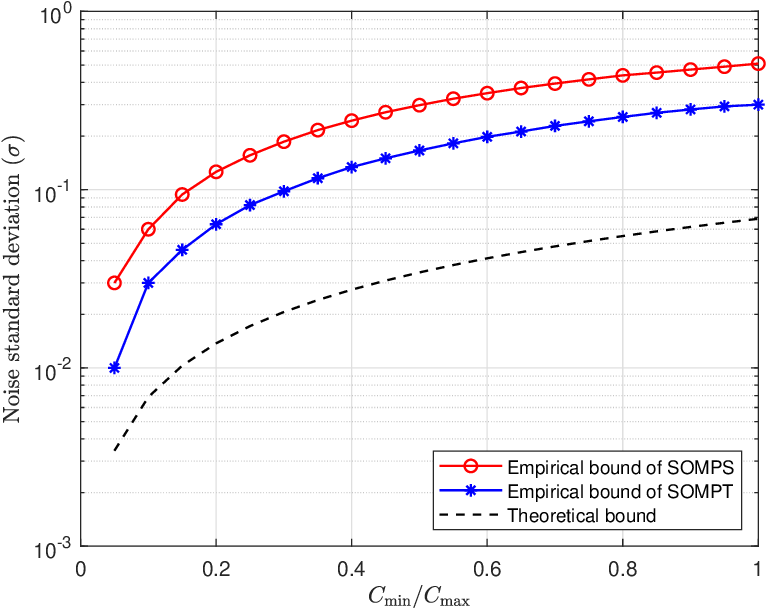}
\caption{Noise level for guarantee of SOMP  versus $C_{\min}/C_{\max}$ under Gaussian noise ($M=100, N=200,  L=4, d=4,C_{\max}=4,[\bN ]_{i,j} \sim \cN(0,\sigma^2),\delta =10^{-3}$).} \label{figure_DR_Gaussian}
\end{figure}

\section{Conclusion} \label{SConclusion}
In this paper, we have analyzed the performance guarantee of SOMP based on the MIP of measurement matrices when the observations are corrupted by noise.
The MIP reveals its amenability and has been widely exploited in various signal processing problems, thoroughly understanding the noisy SOMP in terms of the MIP is emerging.
Specifically, when the noise is bounded, we have shown that if the $\ell_2$-norm of non-zero rows of the row-sparse matrix is lower bounded by ${2 \| \bN \|_2}/({1-(2L-1)\mu })$, the successful recovery of the support set is guaranteed.
On the other hand, when the noise is unbounded, the closed-form lower bound of the SRP was derived.
Based on the derived lower bound, we have shown the conditions for the number of measurements, noise level, the number of sparse vectors, and mutual coherence, on which the required  SRP can be achieved.
Finally, the simulation results validated our analysis, where the proposed bound under $\ell_2$-norm of the noise is tighter than the existing bound \cite{Wang2013Performance}.

\bibliographystyle{IEEEtran}

\bibliography{IEEEabrv,Conference_mmWave_CS}

\clearpage

\end{document}